\documentclass[acmsmall,screen,nonacm]{acmart}
\usepackage{csquotes}
\usepackage{framed}
\usepackage{mathpartir}

\usepackage{enumitem}
\newlist{enumerate*}{enumerate*}{1}
\setlist[enumerate*]{label=(\roman*)}

\usepackage[Tol,pgf]{colorblind}

\usepackage{tikz}
\usetikzlibrary{arrows.meta}
\usetikzlibrary {shapes.multipart} 
\usetikzlibrary {bending,positioning} 
\usetikzlibrary{calc}
\usetikzlibrary{math}

\tikzset{
  dot/.style={
          circle,          %
          fill,         %
          inner sep=1pt     %
  },
  myarrow/.style={
    -Latex
  },
}
\newcommand\drawline[4]{%
\draw (#3) node[left] (#2){#1} edge +(#4,0)
}
\newcommand\labelaction[3]{%
\draw (#1) node[#2, text width=2.6cm, align=center, outer sep=1mm]{#3}
}

\usepackage{subcaption}

\usepackage[figure, linesnumbered]{algorithm2e}
\DontPrintSemicolon
\SetArgSty{sffont}
\SetProcArgSty{sffont}
\SetProcNameSty{sffont}
\SetFuncSty{emph}
\SetFuncArgSty{sffont}

\usepackage[capitalize, noabbrev]{cleveref}
\Crefname{lstlisting}{Figure}{Figures}
\crefname{lstlisting}{Figure}{Figures}
\Crefname{definition}{Definition}{Definitions}
\crefname{definition}{Definition}{Definitions}
\Crefname{theorem}{Definition}{Definitions}
\crefname{theorem}{Definition}{Definitions}
\usepackage{listings}
\lstdefinelanguage{myscala}{%
  language     = Scala,
  morekeywords = {then, enum, given},
}
\makeatletter
\AtBeginDocument{%
  \let\c@figure\c@lstlisting
  
  \let\ftype@lstlisting\ftype@figure %
}
\makeatother

\setcopyright{acmlicensed}
\copyrightyear{2018}
\acmYear{2018}
\acmDOI{XXXXXXX.XXXXXXX}

\acmJournal{JACM}
\acmVolume{37}
\acmNumber{4}
\acmArticle{111}
\acmMonth{8}

\RequirePackage[normalem]{ulem} %
\RequirePackage{color}\definecolor{RED}{rgb}{1,0,0}\definecolor{BLUE}{rgb}{0,0,1} %

\begin{document}

\title{PRDTs: Composable Design and Verification of Consensus Protocols using Replicated Data Types} 
\author{Julian Haas}
\orcid{0000-0001-9959-5099}
\affiliation{%
  \institution{Technische Universität Darmstadt}
  \country{Germany}
}

\author{Ragnar Mogk}
\orcid{0000-0003-4583-1791}
\affiliation{%
  \institution{Technische Universität Darmstadt}
  \country{Germany}
}

\author{Annette Bieniusa}
\orcid{0000-0002-1654-6118}
\affiliation{%
  \institution{Rheinland-Pfälzische Technische Universität
Kaiserslautern-Landau}
  \country{Germany}
}

\author{Mira Mezini}
\orcid{0000-0001-6563-7537}
\affiliation{%
  \institution{Technische Universität Darmstadt}
  \country{Germany}
}
\affiliation{%
  \institution{National Research Center for Applied
Cybersecurity (ATHENE)}
  \country{Germany}
}
\affiliation{%
  \institution{Hessian Center for Artificial Intelligence (hessian.AI)}
  \country{Germany}
}

\begin{abstract}
Consensus protocols are fundamental in distributed systems as they
enable services with strong consistency properties. 
However, designing protocols optimized for specific use-cases under certain system assumptions is typically an error-prone process requiring expert knowledge. Furthermore, while most recent optimized protocols are variations of well-known algorithms like Paxos or Raft, they often necessitate complete re-implementations, potentially introducing new bugs and complicating the application of existing verification results.
This approach impedes application-specific consistency protocols that can easily be amended or swapped out, depending on the given application and deployment scenario.

We propose \emph{Protocol Replicated Data Types} (PRDTs), a novel programming model for implementing consensus protocols using replicated data types (RDTs). Inspired by the knowledge-based view of consensus, PRDTs employ RDTs to monotonically accumulate \emph{knowledge} until \emph{agreement} is reached. This approach allows for implementations focusing on high-level protocol logic that abstracts away network details and facilitates automated verification.
Moreover, by applying existing algebraic composition techniques for RDTs in the PRDT context, we enable composable protocol building-blocks for implementing complex protocols.
We present a formal model of our approach and implement a proof procedure that allows automated reasoning about the consensus safety of concrete PRDT implementations.
Additionally, we demonstrate the applicability of our model in verified PRDT-based implementations of existing consensus protocols, and report empirical performance evaluation results. Our findings indicate that the PRDT approach offers enhanced flexibility and composability in protocol design, facilitates reasoning about correctness, 
and is suited for real-world adoption without intrinsic performance drawbacks.
\end{abstract}

\maketitle

\newcommand\process[2]{(i_{#1}, s_{#2})}
\newcommand\decision[1]{d(#1)}
\newcommand\processes{P}
\newcommand\action{\textit{ProtocolAction}}

\section{Introduction}\label{intro}

In distributed systems, consensus is the foundation of strong consistency, creating order in a distributed environment and
ensuring the correctness and reliability of 
crucial
distributed services.
However, establishing consensus often presents a performance bottleneck~\cite{Brewer2012}, 
which has led to an active area of research developing custom protocols tailored to specific 
scenarios and application requirements~\cite{jhaDerechoFast2019, huntZooKeeperWaitfree2010, kotlaZyzzyvaSpeculative2010, maiyyaUnifyingConsensus2019, maoMenciusBuilding2008, nawabDPaxosManaging2018, ongaroSearchUnderstandable2014, vanrenessePaxosMade2015, weiVCorfuCloudScale2017}.

Developing correct custom (application-specific) consensus protocols is challenging.
Existing programming and verification approaches for consensus
force designers to reason about networks of concurrently operating processes that might fail or suffer from network partitions~\cite{wilcoxVerdiFramework2015,mcmillanDeductiveVerification2018, mcmillanIvyMultimodal2020,hawblitzelIronFleetProving2015}.
Due to the inherent complexity of consensus, new designs are often 
variations of well-established and well-studied protocols~\cite{padonPaxosMade2017, wangParallelsPaxos2019, chandraPaxosMade2007, howardRaftRefloated2015, lamportGeneralizedConsensus2005, vanrenessePaxosMade2015}.
Despite being structurally similar to established designs, such variations are usually realized as distinct, monolithic implementations with subtle differences in system assumptions and use cases.
This hinders the transfer of insights between protocols and increases the risk of introducing new bugs.
Meanwhile, research on the \enquote{building blocks} that constitute a consensus protocol has only recently received attention~\cite{andersenProtocolCombinators2021, howardGeneralisedSolution2019, howardPaxosVs2020, guptaChemistryAgreement2023,jaberQuickSilverModeling2021}.

To improve performance and simplify reasoning, 
alternative knowledge-based approaches to the consensus problem
have recently gained popularity.
They follow the framing of 
\emph{agreement} as a problem of knowledge by \citet{halpernKnowledgeCommon1990}:
To reach agreement (and thus consensus), processes have to establish \emph{eventual common knowledge},
denoting facts that all processes will know eventually, and that remain valid thereafter.
Through the exchange of facts, processes can try to lift local knowledge
to eventual common knowledge in a monotonic process.
The knowledge-based model improves performance and scalability of consensus protocols~\cite{jhaDerechoFast2019, chuOptimizingDistributed2024}, while simplifying their formal safety verification~\cite{howardGeneralisedSolution2019, nagasamudramVerifyingImplementation2024,lewchenkoBoltOnStrong2025}.

Motivated by these results, our work introduces \emph{protocol replicated data types (PRDTs)}: a novel programming model that uses algebraic convergent replicated data types to enable developers to implement, compose, and customize consensus
protocols for diverse systems and applications.
PRDTs store the local knowledge of each process in a replicated data type, which ensures that the result of every local action becomes eventual common knowledge.
Given this, the remaining responsibility of a PRDT implementation is to ensure that \emph{local} actions can never create new knowledge that would invalidate existing decisions.
Limiting the reasoning to local actions drastically simplifies reasoning about the correctness of a consensus implementation, as was recently demonstrated by \citet{lewchenkoBoltOnStrong2025} in a proof theory that was developed in parallel to our work on PRDTs.
In this paper, we adapt their proof theory to our model and are able to show that PRDTs guarantee two of their proof obligations by design, further reducing the proof burden.

PRDTs build upon algebraic replicated data types (ARDTs)~\cite{kuessnerAlgebraicReplicatedData2023}
and introduce specialized abstractions for systematically building safe consensus protocols on top of the convergent replicated state captured by the ARDT framework.
ARDTs allow algebraic composition of existing replicated data types, and we demonstrate that this directly extends to PRDTs, fostering \emph{composable and parametric consensus protocols}: developers can modularly assemble consensus protocols from provably correct components.
This modular programming model encourages sharing and reuse of building blocks across protocols, reducing the risk of introducing subtle errors when adapting consensus mechanisms to new domains, while supporting \emph{compositional reasoning and verification}.
As a case study, we provide \emph{several executable implementations} of distributed protocols from the Paxos family that we derive by composing simpler components, showcasing how PRDTs enable modular design and integrate with existing language features such as parametric algebraic data types, traits, and type classes (\cref{implementing-a-real-protocol}).

Finally, we conduct a \emph{performance evaluation}\footnote{\label{artifact}We plan to submit our performance evaluation and our PRDT implementations, verified with Stainless, for artifact evaluation. All of our implementations are FOSS under Apache-2.0 licence and will be made available with the non-anonymized version of this work.} to empirically assess the practical feasibility of the PRDT model and the potential overhead they might introduce compared to existing message-based consensus protocol implementations (\cref{sec:evaluation}).
To this end, we implement a strongly consistent key-value store based on our PRDT implementation of Multi-Paxos and benchmark its performance in both local and geo-replicated settings.
This evaluation provides a proof-of-principle that constructing practical systems using PRDT-based protocol implementations is feasible. Our prototype achieves parity with etcd~\cite{etcd} and sometimes outperforms it in scenarios where the servers are colocated in the same data centre.
Performance differences are visible in geo-replicated setups where etcd profits from optimizations in its key-value store layer (e.g., request pipelining), that we did not implement in our prototypical implementation.

In summary, this paper makes the following contributions:

\newpage

\begin{itemize}
        \item A \textbf{dedicated programming model for knowledge-based compositional consensus protocols}, moving the state of the art from isolated solutions targeted at specific protocols %
        to a general model that systematically supports the creation, modular composition, and flexible customization of a broad range of knowledge-based protocols (\cref{sec:overview,sec:verification}).
       \item A compositional \textbf{automated verification procedure} for our programming model that allows for the verification of concrete PRDT implementations. We implement the procedure using Stainless~\cite{epfllaraStainlesFormal}, a verifier for Scala programs.
       Our verification process leverages the key features of PRDTs -- replicated state convergence and algebraic composability -- to reduce the burden of manual verification, simplifying and automating existing theories (\cref{sec:verification}).
        \item \textbf{Several executable implementations of  protocols from the Paxos family}\textsuperscript{\ref{artifact}} derived by composing simpler components, paving the way towards providing developers -- of system and application software alike -- with libraries of reusable provably safe protocol modules. (\cref{sec:implementation})
\end{itemize}

\section{Overview of the PRDT Programming Model}
\label{sec:overview}

A Protocol RDT (PRDT) consists of three components:
\begin{enumerate}
    \item A \emph{knowledge lattice} representing the protocol state, with each element representing a process's current knowledge about the system.
    \item A set of \emph{protocol actions} with corresponding \emph{preconditions}, which guard their execution.
    \item A \emph{decision function} that makes \emph{irrevocable} decisions based on the current knowledge state.
\end{enumerate}

We introduce these components using a simple voting protocol for illustration. 
In this protocol, three processes, $a$, $b$, and $c$, attempt to agree on a type of pet through a single round of voting. 
Each process can vote once, and the pet type that receives a majority of votes (two in this case) is considered the agreed-upon value.
\Cref{fig:voting-scenario} illustrates an example execution of the voting protocol. Here, both $a$ and $b$ vote for \emph{cat}. 
As soon as $b$ is notified about $a$'s choice, it can make a local decision. 
This decision then propagates to the other two processes once they receive $b$'s vote.
A concrete Scala implementation of the example voting protocol as a PRDT is given in \cref{fig:voting-rdt}.

\begin{figure}
\hspace{-1cm}
\centering%
  \begin{tikzpicture}[every node/.style={font=\itshape\tiny}, scale=0.5]%
    \input{figures/message-diagram-rdt-pet.tex}%
  \end{tikzpicture}%
\caption{Example run of the Voting protocol.}%
\label{fig:voting-scenario}%
\end{figure}

\cref{fig:value-lattice} visualizes the knowledge lattice and decisions for two distinct runs of the voting protocol that reach different decisions.
On the left and right sides, we depict all possible states in the knowledge lattice of the given run. 
Solid arrows depict the partial order of the knowledge lattice; 
they correspond to permissible system transitions induced by protocol actions or by \emph{merging} one state into another.
In the middle, we depict \emph{decisions}, the possible results of the decision function.
Dashed arrows show the result of the decision function for a given knowledge state. 
Decisions have a partial order, also visualized as solid arrows.

\begin{figure}
\begin{lstlisting}[language=myscala, basicstyle=\scriptsize\ttfamily, xleftmargin=.1\textwidth, xrightmargin=.1\textwidth, escapeinside={//§}{\^^M}, caption={Voting as a PRDT.}, captionpos=b, label={fig:voting-rdt}]
val participants: Set[Uid] = ... // fixed for all process //§\label{line:voting-participants}

case class Voting[A](votes: Map[Uid, A]) { // knowledge state //§\label{line:voting-state}

  // preconditions
  def hasNotVoted(using LocalUid): Boolean = //§\label{line:voting-hasvoted}
    !votes.contains(replicaId)

  // protocol actions
  def voteFor(value: A)(using LocalUid): Voting[A] = //§\label{line:voting-votefor}
    precondition(hasNotVoted) { //§\label{line:updateif}
      Voting(Map(replicaId -> value)) //§\label{line:voting-voteforreturn}
    }

  // decision function
  def decision: Decision[A] = //§\label{line:voting-decision}
    getLeadingValue() match
      case Some(value, count) if count > participants.size / 2 => Decided(value)
      case _ => Undecided
}

enum Decision[A] { //§\label{line:voting-decision-domain}
  case Undecided
  case Decided(value: A)
}
\end{lstlisting}
\end{figure}

\begin{figure}
\begin{tikzpicture}[every node/.style={font=\itshape\tiny}, scale=1.5]
    \node[minimum width=1.3cm,, align=center] (A) at (1,-0.5) {undecided};
    \node[minimum width=1.3cm,, align=center] (B) at (0.5,1.0) {decided: cat};
    \node[minimum width=1.3cm,, align=center] (E) at (1.5,1.0) {decided: dog};

    \node[minimum width=1.3cm,, align=center] (F) at (-2,-1) {$\emptyset$};
    \node[minimum width=1.3cm,, align=center] (G) at (-1,0) {$\{a \mapsto \textit{cat}\}$};
    \node[minimum width=1.3cm,, align=center] (H) at (-2,0) {$\{b \mapsto \textit{cat}\}$};
    \node[minimum width=1.3cm,, align=center] (I) at (-3,0) {$\{c \mapsto \textit{dog}\}$};
    \node[minimum width=1.3cm,, align=center] (J) at (-1,1) {$\{a \mapsto \textit{cat},$ \\ $b \mapsto \textit{cat}\}$};
    \node[minimum width=1.3cm,, align=center] (b_cat_c_dog) at (-3,1) {$\{b \mapsto \textit{cat},$ \\ $c \mapsto \textit{dog}\}$}; %
    \node[minimum width=1.3cm,, align=center] (a_cat_c_dog) at (-2,1) {$\{a \mapsto \textit{cat},$ \\ $c \mapsto \textit{dog}\}$}; %
    \node[minimum width=1.3cm,, align=center] (K) at (-2,2) {$\{a \mapsto \textit{cat}, b \mapsto \textit{cat}, c \mapsto \textit{dog}\}$};

    \node[minimum width=1.3cm,, align=center] (L) at (4,-1) {$\emptyset$};
    \node[minimum width=1.3cm,, align=center] (M) at (3,0) {$\{a \mapsto \textit{dog}\}$};
    \node[minimum width=1.3cm,, align=center] (N) at (4,0) {$\{b \mapsto \textit{cat}\}$};
    \node[minimum width=1.3cm,, align=center] (O) at (5,0) {$\{c \mapsto \textit{dog}\}$};
    \node[minimum width=1.3cm,, align=center] (b_cat_c_dog_2) at (5,1) {$\{b \mapsto \textit{cat},$ \\ $c \mapsto \textit{dog}\}$}; %
    \node[minimum width=1.3cm,, align=center] (a_dog_c_dog) at (4,1) {$\{a \mapsto \textit{dog},$ \\ $c \mapsto \textit{dog}\}$}; %
    \node[minimum width=1.3cm,, align=center] (a_dog_b_cat) at (3,1) {$\{a \mapsto \textit{dog},$ \\ $b \mapsto \textit{cat}\}$};
    \node[minimum width=1.3cm,, align=center] (Q) at (4,2) {$\{a \mapsto \textit{dog}, b \mapsto \textit{cat}, c \mapsto \textit{dog}\}$};
    \path [->] (L) edge (M)
                   edge (N)
                   edge (O);
    \path [<-] (a_dog_b_cat) edge (M)
                             edge (N);
    \path [<-] (a_dog_c_dog) edge (M)
                             edge (O);
    \path [<-] (b_cat_c_dog_2) edge (N)
                               edge (O);
    \path [<-] (Q) edge (a_dog_c_dog)
                   edge (b_cat_c_dog_2)
                   edge (a_dog_b_cat);
    \path [<-] (A) edge[dashed] (M)
                   edge[dashed] (N)
                   edge[dashed] (L)
                   edge[dashed] (O);
    \path [->] (Q) edge[dashed] (E);

    \path [->] (F) edge (G)
                   edge (H)
                   edge (I);
    \path [<-] (J) edge (G)
                   edge (H);
    \path [<-] (b_cat_c_dog) edge (H) %
                           edge (I); %
    \path [<-] (a_cat_c_dog) edge (G) %
                                   edge (I); %
    \path [<-] (K) edge (b_cat_c_dog)
                   edge (J)
                   edge (a_cat_c_dog);
    \path [->] (A) edge (B)
                   edge (E);
    \path [<-] (A) edge[dashed] (F)
                   edge[dashed] (G)
                   edge[dashed] (H)
                   edge[dashed] (I)
                   edge[dashed] (b_cat_c_dog)
                   edge[dashed] (b_cat_c_dog_2)
                   edge[dashed] (a_dog_b_cat)
                   edge[dashed] (a_cat_c_dog);
    \path [<-] (B) edge[dashed] (J)
                   edge[dashed] (K);
    \path [<-] (E) edge[dashed, bend left=50] (a_dog_c_dog);
\end{tikzpicture}
  \caption{\label{fig:value-lattice} The knowledge lattices for two different runs of the voting protocol (left and right). The respective decision is displayed in the middle with dashed arrows pointing to the result of the decision function.}
\end{figure}

\subsection{Knowledge Lattice}

\begin{definition}[Knowledge Lattice]\label{def:knowledge-lattice}
The knowledge lattice $(S, \sqcup)$ of a PRDT is a join-semilattice with a partially ordered set $S$, and a merge function $\sqcup: S \times S \rightarrow S$ that is idempotent, commutative, and associative. As for any lattice, the partial order and merge function on elements $s_1, s_2 \in S$ imply each other: $s_1 \leq s_2 \iff s_1 \sqcup s_2 = s_2$.
\end{definition}

We implement the knowledge lattices using \emph{Algebraic Replicated Data Types (ARDTs)}~\cite{kuessnerAlgebraicReplicatedData2023}.
Like CRDTs~\cite{preguica2018, Shapiro2011}, ARDTs guarantee Strong Eventual Consistency (SEC) -- due to the underlying semi-lattice theory
the merge function composing replicas of the same type is robust to message duplication and reorder.
But ARDTs extend beyond CRDTs: They enable the modular construction of complex
replicated types from simpler ones via algebraic composition operators: the product types for combinations and the sum types for alternatives.
The system automatically derives the merge function of any composite
ARDT from the merge function of its constituents.
For product types, the generated merge function merges all product components individually, and for sum types, the merge function orders the alternatives by their ordinal value.
The automatic derivation of correct merge functions for complex ARDTs 
enables developers to model complex knowledge domains.
It is crucial for ensuring the correctness and compositionality of protocols in our model.

Each process maintains a \emph{local knowledge state}, which is an element of the PRDT knowledge lattice. 
The effects of local protocol actions by the process and remote updates synchronized from other processes are merged into the local state.
Modelling protocol knowledge as a lattice ensures that information is never lost: remote updates merge without conflict and strictly grow the local knowledge state. This monotonic growth ensures that the local states of all processes converge to the same common knowledge state, thus ensuring SEC, i.e., all processes eventually reach the same state, assuming that all updates are eventually delivered.

For the voting protocol, the knowledge to track consists of the votes per process.
We model this knowledge in the \texttt{Voting} PRDT using a \emph{Map} 
ARDT -- a product algebraic type that associates process IDs with voted-for values 
(\cref{fig:voting-rdt}, Line~\ref{line:voting-state}).
Note that \texttt{Voting} relies on a list of unique IDs (Uid) of the participating processes (Line~\ref{line:voting-participants}).
This list is fixed and is used to determine the majority needed for a decision.

\subsection{Protocol Actions and Their Guarding Preconditions}

\begin{definition}[Protocol Action]\label{def:protocol-action}
A \emph{protocol action} is a function $a: \textit{Uid} \rightarrow S$ that for a given unique process identifier $i \in \textit{Uid}$ returns a new \emph{delta} in $S$.
Each action $a$ has a \emph{precondition} $p_a \subseteq \textit{Uid} \times S$, which denotes the subset of $\textit{Uid} \times S$ for which $a$ is enabled.
\end{definition}

\emph{Protocol actions} specify how to add knowledge to the knowledge lattice. 
Executing a protocol action produces a \emph{delta} state\footnote{The use of deltas makes PRDTs a variant of delta state-based replicated data types~\cite{Almeida2018}.}.
Deltas are updates in the form of states that can be merged into copies of this data type to grow its knowledge.
To express protocol-specific requirements about permissible actions,
protocol actions can be equipped with \emph{preconditions}, 
which guard their execution based on the knowledge state and the process ID. 

The \emph{Voting} protocol has a single action \texttt{voteFor} (Line~\ref{line:voting-votefor}), which lets processes vote for a given value using their process ID\footnote{
  The \texttt{using LocalUid} is a context parameter in Scala. It provides access to the \texttt{replicaId} function, which returns the \texttt{Uid} of the current process.
}.
Deltas in our example wrap a single vote; for example, \texttt{Voting(Map(a -> cat))} is a delta that expresses that $a$ has voted for $cat$. 
A correct implementation of this voting protocol must ensure that every process only votes once, e.g., that \texttt{Voting(Map(a -> cat))} and \texttt{Voting(Map(a -> dog))} may never exist in the context of the same protocol execution.
To express protocol-specific requirement, the action \texttt{voteFor} has precondition \texttt{hasNotVoted} (Line~\ref{line:voting-hasvoted}), which checks if there already exists a vote for the local process ID.\footnote{
Note that in the formal representation we omit the parameters that actions have in the implementation.
We can model parameters without loss of generality by treating protocol actions with different parameters as different actions, e.g., \texttt{voteFor(cat)} and \texttt{voteFor(dog)} are modelled as two different protocol actions.}

To ensure correct and deterministic protocol decisions, preconditions must not depend 
on values that could be changed by concurrent protocol actions by other processes. 
We say that they need to be \emph{stable} upon remote updates. 
For example, checking that there is no vote for the local process ID is a stable precondition
because in this implementation, no other process can add a vote for the local process ID.
On the contrary, checking that there are no votes at all would not be a stable precondition, 
because other processes can add votes concurrently.
The \emph{stability} property is crucial for enabling local reasoning 
about the safety of a given PRDT implementation (\cref{sec:verification}).

\subsection{Decision Function and Decision Domain}
\label{sec:decision-function-overview}

\newcommand{\dec}{\mathit{dec}}

\begin{definition}[Decision Function]\label{def:decision}
Given the knowledge lattice $S$ and a partially ordered set of decisions $D$,
a \emph{decision function} $\dec{}: S \rightarrow D$ maps elements in $S$ to elements in $D$.
\end{definition}

The \emph{decision function} of a PRDT takes the current knowledge state and yields a value in the decision domain $D$.
The partial order on $D$ determines which decisions are 
compatible with each other. A transition from decision $d_1$ to decision $d_2$, where $d_1, d_2 \in D$  is permitted only if $d_1 \leq d_2$.
We say the larger decision is \emph{compatible} with the smaller one. If two decisions are incomparable (unrelated by $D$'s partial order), they are considered incompatible, meaning no transitions are allowed between them.
The PRDT verification procedure (Section \ref{sec:verification}) guarantees two properties that are key for consensus: first, that every process monotonically advances its local decision state -- transitioning only to compatible decisions -- and, second, that there always exists a decision that is compatible with the current decision of all processes.
By choosing a decision domain, PRDT developers can express various types of consensus, including classic protocols like single-shot (deciding a single value once) and multi-shot (continuously deciding on a series of values) consensus.

The decision function of \texttt{Voting} (Fig. \ref{fig:voting-rdt}, Line~\ref{line:voting-decision}) determines whether a certain value has received the majority of votes. If so, it returns this value as the decided value (e.g., \texttt{Decided(dog)}); otherwise, it returns \texttt{Undecided}.
The partial order on decisions entails that \texttt{Undecided} $\leq$ \texttt{Decided(v)} for any value $v$, allowing transitions from an undecided state to a decided one. Distinct decisions are incompatible (e.g., \texttt{Decided(cat)} $\not\leq$ \texttt{Decided(dog)}).
Consequently, any \texttt{Decided(v)} serves as a final state since there are no strictly greater elements in the order. The only permissible transition is from \texttt{Undecided} to a specific \texttt{Decided} value.

As an example with multiple compatible decisions, consider a global append-only log.
In the decision domain, each decision is a sequence of values, such as \emph{Log(a)} or \emph{Log(a, b, c)}.
We define the decision order by the prefix relation: \texttt{Log(a)} $\leq$ \texttt{Log(a,b,c)} because \texttt{Log(a)} is a prefix of \texttt{Log(a,b,c)}.
Under this partial order, it is permissible for a process to append an element, for example, by moving from \texttt{Log(a)} to \texttt{Log(a,b)}, since this transition grows the decision. In contrast, it is not permissible for two processes to independently decide \texttt{Log(a,b)} and \texttt{Log(a,c)} because there is no decision that is compatible with (i.e., greater than) both outcomes. As a result, the two processes could not be reconciled. We formalize the intuition behind this in the consensus safety property in \cref{sec:consensus-properties}.

\subsection{Execution Semantics of PRDT Systems}\label{sec:semantics-systems}
A \emph{PRDT system} is a distributed system executing a protocol defined by a PRDT. 
We model the execution behaviour of a PRDT system as a system knowledge and two rules that govern 
transitions in the knowledge.

\begin{definition}[System Knowledge]
Let $\textit{Uid}$ be a set of process identifiers and $S$ be the PRDT \emph{knowledge lattice}.
The system knowledge $K: \textit{Uid} \rightarrow \mathcal{P}(S)$ maps each process identifier to its current local knowledge, represented as a set of deltas $s \subseteq S$.
We write $K[i \mapsto s]$ to express that $K$ associates the set of deltas $s$ with process $i$.
\end{definition}

The system knowledge tracks the global state by maintaining the individual deltas observed by each process.
This representation offers finer granularity than tracking only the merged state,
but the current effective state of process $i$ can always be reconstructed by joining all (locally) known deltas: $\bigsqcup K(i)$ (cf. \cref{def:knowledge-lattice}).

\paragraph{Transition Rules}
Two transition rules govern the behaviour of a PRDT system: \textsc{apply} for local computation and \textsc{merge} for communication (see \cref{fig:rules}).
Given that the precondition $p_a(i, \bigsqcup K(i))$ for protocol action $a$ holds at process $i$, the application of the \textsc{apply} rule produces a step in the system knowledge $K \rightarrow K'$ by merging the delta $a(i)$ into the current knowledge of $i$, $K(i)$. The \textsc{merge} rule
models information propagation, allowing process $i$ to incorporate deltas known by process $j$.
By merging all deltas, the merge rule ensures that deltas are delivered in causal order.
The \textsc{apply} and \textsc{merge} rules generate a sequence of atomic steps, $K \rightarrow K'$, defining the global system trace.

\begin{figure}
\begin{mathpar}
\inferrule*[Left=apply]{%
p_a(i, s_i)
}{%
K[i \mapsto s_i] \rightarrow K[i \mapsto s_i \cup \{a(i)\}]
}
\and
\inferrule*[Left=merge]{
~
}{
K \rightarrow K[i \mapsto K(i) \cup K(j)]
}
\end{mathpar}
\caption{System transition rules.}
\label{fig:rules}
\end{figure}

\paragraph{Partial Failures}
The model allows for process crashes and network partitions.
A process crash corresponds to an execution trace where the \textsc{apply} and \textsc{merge} rules cease for a specific process.
Similarly, a network partition corresponds to a trace where the \textsc{merge} rule is never invoked between a specific pair of processes.
Since our PRDT correctness verification (Sec. \ref{sec:verification}) holds for \emph{all} valid execution traces generated by the rules, it also holds for traces exhibiting these fault patterns.

\subsection{Composition of PRDT-Based Protocols}
\label{sec:composing-rdts}

The algebraic composition 
supported by ARDTs enables algebraic protocol composition with PRDTs.
For illustration, consider a scenario that requires simultaneous voting on two values.
We can model this using two instances of the \texttt{Voting} PRDT from \cref{fig:voting-rdt}:
\begin{lstlisting}[language=myscala, basicstyle=\scriptsize\ttfamily, xleftmargin=.1\textwidth, xrightmargin=.1\textwidth, escapeinside={//§}{\^^M}, label={fig:parallelvoting-decision}]
case class ParallelVoting[A,B](voting1: Voting[A], voting2: Voting[B]) {
  def decision[A] = //§\label{line:parallelvoting-decision}
    (voting1, voting2) match
      case (Decided(a), Decided(b)) => Decided((a,b))
      case _ => Undecided
}
\end{lstlisting}
The merge function of the knowledge lattice for \texttt{ParallelVoting} automatically 
derives from the merges of the two \texttt{Voting} instances.
Additionally, \texttt{ParallelVoting} can make use of the protocol actions defined for its components, 
with each such action only modifying the state of the respective component.
Using composition ensures that the invariants enforced by the actions of each component are guaranteed, as long as 
no new protocol actions are defined that modify the components.
However, developers can extend composite PRDTs such as \texttt{ParallelVoting}
with new protocol actions and decision functions as is done in Line~\ref{line:parallelvoting-decision}:
This specialized decision function for \texttt{ParallelVoting} only makes a decision 
when both voting instances have made individual decisions.
This illustrates how composite PRDTs allow developers to define novel decision logic that orchestrates the outcomes of their constituent components, thereby enabling more sophisticated coordination behaviours than individual PRDTs on their own. 
This allows for subtle adjustments to coordination logic.
In \cref{sec:implementation}, we will elaborate more on how to define complex PRDTs such as the Paxos PRDT 
by composing simpler PRDTs such as the voting PRDT.

\section{Verifying Correctness of PRDT-based Protocol Implementations}
\label{sec:verification}

In this section, we formalize the correctness properties of the PRDT programming model. We begin by defining \emph{consensus safety} for PRDTs (Sec. \ref{sec:consensus-properties}) and discuss how it relates to classic consensus guarantees (Sec. \ref{sec:classical-consensus-properties}).
Finally, we derive a set of proof obligations sufficient to establish consensus safety (Sec. \ref{sec:proof-obligations}) and, after illustrating them by manually verifying the correctness of the Voting example protocol (Sec. \ref{sec-verifying-voting}), we demonstrate how to automate their verification (Sec. \ref{sec-automating-verification}) using the Stainless verifier for Scala~\cite{epfllaraStainlesFormal}.

\subsection{Consensus Safety of PRDTs}\label{sec:consensus-properties}

\begin{definition}[Consensus safety]\label{def:agreement}\label{def:consensus-safety}
A PRDT is consensus-safe if for any reachable knowledge state $K$ and any process identifiers $i, j \in \textit{Uid}$, the following holds:
\begin{align*}
(1) \quad &\exists d \in D.\ \  \dec(\textstyle\bigsqcup K(i)) \leq d \land \dec(\textstyle\bigsqcup K(j)) \leq d  \\
(2) \quad &K \rightarrow \dots \rightarrow K' \implies \dec(\textstyle\bigsqcup K(i)) \leq \dec(\textstyle\bigsqcup K'(i))
\end{align*}
\end{definition}

Condition (1) ensures that the decisions of any two processes remain reconcilable: there exists a decision $d$ that is compatible with both. This prevents the system from reaching a state where two processes hold incompatible decisions with no path to convergence.
Condition (2) enforces monotonicity of the decision function: as a process gains knowledge, its decision values can only increase according to the order of the decision domain. This prevents a process from reverting to an earlier decision or switching to an incompatible one.

\subsection{Consensus Safety and Standard Consensus Properties}\label{sec:classical-consensus-properties}

In the following, we relate our consensus safety (\cref{def:agreement}) to classical consensus properties.

\subsubsection{Correctness of Single-shot Consensus Protocols}

\citet[Module~5.1]{cachinIntroductionReliable2011} state the following correctness properties for single-shot consensus protocols.
\begin{enumerate}
    \item Agreement: No two correct processes decide differently.
    \item Integrity: No process decides twice.
    \item Termination: Every correct process eventually decides some value.
    \item Validity: If a process decides $v$, then $v$ was proposed by some process.
\end{enumerate}

Agreement and integrity are \emph{safety} guarantees, without them a client of the protocol could observe incompatible decisions.
In the following, we prove that consensus safety (\cref{def:agreement}) implies both agreement and integrity properties when using a decision function with the single-shot consensus domain from our voting example.
In that domain, there is a single undecided value and multiple individual decisions. Each decision is greater than the undecided value but incomparable to other decisions, making each decision maximal according to the decision order.

\begin{theorem}
Consensus safety implies agreement.
\end{theorem}
\begin{proof}
We prove the theorem by contradiction. Assume two processes $i$ and $j$ violate agreement, each deciding on different values $d_i = \dec(\bigsqcup K(i))$ and $d_j = \dec(\bigsqcup K(j))$, where both $d_i$ and $d_j$ are decided (maximal) but not equal: $d_1 \neq d_j$. Due to condition (1) of consensus safety (\cref{def:agreement}), there exists a decision $d$ such that $d_i \leq d$ and $d_j \leq d$. Since all decisions in a single-shot domain are maximal elements of the decision order, any $d$ greater than a maximal decision must be equal to that decision. Therefore, $d_i = d = d_j$, contradicting the assumption. Thus, consensus safety implies agreement.
\end{proof}

\begin{theorem}
Consensus safety implies integrity.
\end{theorem}
\begin{proof}
We prove the theorem by contradiction.
Assume process $i$ violates integrity by deciding twice. 
Then, there exists a trace $K \rightarrow \dots \rightarrow K'$, where $d = \dec(\bigsqcup K(i))$ and $d' = \dec(\bigsqcup K'(i))$ 
are both decided (maximal) decisions, but $d \neq d'$. Due to condition (2) of consensus safety (\cref{def:agreement}) monotonicity ensures $d \leq d'$.
Since both are maximal elements in the single-shot decision domain, this implies that $d = d'$, contradicting the assumption. 
Thus, consensus safety implies integrity.
\end{proof}

\subsubsection{Correctness of Multi-shot Consensus Protocols}

Beyond single-shot consensus, \citet[Module~6.12]{cachinIntroductionReliable2011} define correctness properties for more expressive consensus variants, which PRDTs also support. A canonical example is the replicated state machine model, where each process maintains a log of committed outputs. The decision domain of multi-shot consensus are sequences of elements such as \texttt{Log(a)} or \texttt{Log(a, b, c)}, 
where a decision $d$ is $\leq d'$ if $d$ is a prefix of $d'$.
According to \citet[Module~6.12]{cachinIntroductionReliable2011}, multi-shot consensus protocol implementations must satisfy the following properties:

\begin{enumerate}
    \item Multi-Shot Agreement: All correct processes obtain the same sequence of outputs\footnote{
    We assume that output here refers to decisions that cannot change, i.e. integrity of outputs is implied. This formulation of agreement also seems to imply aspects of termination, agreement should only require that processes agree on the observed outputs, with some processes lagging behind.
    }.
    \item Multi-Shot Termination: If a correct process executes a command, then the command eventually produces an output.
\end{enumerate}

\begin{theorem}
Consensus safety implies multi-shot agreement.
\end{theorem}
\begin{proof}
We prove the theorem by contradiction.
Assume an agreement violation of processes $i, j$ with states $K(i), K(j)$, and decisions $d_i = \dec(\bigsqcup K(i)), d_j = \dec(\bigsqcup K(j))$, where $d_i$ and $d_j$ where at least one element of $d_i$ differs from $d_j$
However, part (1) of consensus safety (\cref{def:agreement}) guarantees the existence of a decision $d$ with $d_i \leq d$ and $d_j \leq d$.
But this implies that both $d_i$ and $d_j$ are prefixes of $d$.
Without loss of generality, assume that $d_i$ is the shorter prefix, then this implies that $d_i \leq d_j \leq d$, but $d_i \leq d_j$ is a contradiction to the assumption that they differ in at least one element, thus consensus safety implies multi-shot agreement.
\end{proof}

\subsubsection{Termination}
Termination (condition 3 of the standard correctness properties for single-shot, and condition 2 in multi-shot consensus protocols) is a \emph{liveness} property (what must eventually hold), which is outside the scope of this work. The classic FLP impossibility result~\cite{fischerImpossibilityDistributed1985} proves that no consensus protocol can guarantee both safety and liveness simultaneously in a fully asynchronous environment subject to failures. Consequently, practical protocols must relax one of these guarantees: they typically aim to preserve safety while ensuring liveness only under synchrony assumptions or bounded fault rates. This work establishes safety guarantees in the general asynchronous setting, independent of fault assumptions.

\subsubsection{Validity}
\label{sec:validity}

Consensus safety does not imply validity (condition 4 of single-shot consensus): While consensus safety prevents incompatible decisions, validity constrains the set of permissible decision values. Crucially, however, establishing validity
does not require global reasoning 
and can be achieved by treating a PRDT as a standard sequential data structure and 
by tracking and analyzing the causal history of actions for a process $i$ and knowledge $K$.
Specifically, we construct and analyze 
$\textit{happenedBefore}(K,i)$ as follows:
when a step $K \rightarrow K'$ applies a local action $a$ at process $i$, we record $a$ in the history of $i$ as $a \in  \textit{happenedBefore}(K', i)$. When a merge step occurs, all actions known at the source process $j$ are added to the history of process $i$.
We then need to specify which action $a$ constitutes proposing a value according to the semantics of the PRDT.
Given the $\textit{happenedBefore}$ relation, when a process $i$ decides $v$, checking whether $v$ was proposed by $i$ corresponds to checking whether there is $a(v) \in \textit{happenedBefore}(K, i)$.

We can verify this for the voting protocol, which uses the \texttt{voteFor(v)} to propose values, by inspecting the code. Specifically, the decision function of \emph{Voting} only produces results that are in the votes map, and the \texttt{voteFor(v)} action is the only method that inserts element $v$ into the map. Thus, the vote $v$ must have been proposed before $decision(v)$.

\subsection{Safety Proof Obligations for the PRDT Model}
\label{sec:proof-obligations}

Directly verifying the \emph{consensus safety} (\cref{def:agreement}) is challenging because it necessitates global reasoning over the combined states of all processes. We circumvent this complexity in two steps. In \cref{sec:step-monotonicity}, we show that consensus safety follows from \emph{monotonicity of decisions for individual steps} of the system.
In \cref{sec:local-to-global}, we demonstrate that local reasoning
is sufficient to guarantee step monotonicity across all transitions, including merges.

\subsubsection{Reducing consensus safety to step monotonicity.}
\label{sec:step-monotonicity}

\begin{definition}[Step Monotonicity]\label{def:stepMono}
A PRDT system is \emph{step-monotonic}, if the decision function is monotone for any step that 
any process $i$ makes:
$$K \rightarrow K' \implies \dec{}(\textstyle \bigsqcup K(i)) \leq \dec{}(\textstyle \bigsqcup K'(i))$$
\end{definition}

\begin{theorem}\label{theorem:step-monot-is-agreement}\label{def:ConsensusFromStepMono}
Consensus safety follows from step monotonicity and strong eventual consistency.
\end{theorem}
\begin{proof}
Assuming step monotonicity, i.e., $K \rightarrow K' \implies \dec{}(\bigsqcup K(i)) \leq \dec{}(\bigsqcup K'(i))$, directly implies condition (2) of \emph{consensus safety} in \cref{def:agreement}, because we can extend the single step to multiple steps due to the transitivity of the decision order.
Strong eventual consistency implies that two processes eventually observe the same deltas, i.e, for any $K$ there exists a sequence of steps 
$K \mapsto^* K'$ 
with $K'(i) = K'(j)$, which implies that $d' = \dec{}(\bigsqcup K'(i)) = \dec{}(\bigsqcup K'(j))$.
Then, from monotonicity of the process steps, it follows that $\dec{}(\bigsqcup K(i)) \leq d'$ and $\dec{}(\bigsqcup K(j)) \leq d'$.
Thus, we have found a $d'$ as required by condition (1) of \emph{consensus safety} in \cref{def:agreement}.
\end{proof}

\subsubsection{From global to local verification conditions}
\label{sec:local-to-global}
While step monotonicity of the system is still a global property, we show that two local verification conditions are sufficient to establish it in the presence of monotonic structure of the knowledge lattice: 
(a) action monotonicity and (b) stability of preconditions.

\begin{definition}[Action Monotonicity]\label{def:action-monotonicity}
Action monotonicity for some protocol action $a$ holds if for any precondition $p_a$, any state $s$, and process $i$:
$$
p_a(s, i) \Rightarrow \dec{}(s) \leq \dec{}\big(s \sqcup a(i)\big)
$$
\end{definition}

\emph{Action monotonicity} guarantees that an application system step (\Cref{fig:rules}) cannot harm the monotonicity of decisions.

\begin{definition}[Stability of Preconditions]\label{def:stability}
A precondition $p_{a_1}$ is \emph{stable} if the following holds for any
concurrent protocol action $a_2$ in state $s$ with process identifiers $i \neq j$:
$$
p_{a_1}(i, s) \land p_{a_2}(j, s) \quad \Rightarrow \quad p_{a_1}(i, s \sqcup a_2(j))
$$

\end{definition}

\emph{Stability} preserves the precondition of concurrent protocol actions executed by different processes in the same state.
We use it to show that every state that is the result of a merge could also be reached by a sequential application of protocol actions.
Hence, monotonicity of the decision function for the merge follows from the monotonicity of each protocol action in the sequence:

\begin{lemma}[Merge Monotonicity]\label{lemma:merge-monotonicity}
Assume a PRDT system with system knowledge $K$ and a decision function $\dec{}$ where action monotonicity holds for all protocol actions and all preconditions are stable.
Let $s_i = \bigsqcup K(i)$ and $s_j = \bigsqcup K(j)$ be the knowledge states of processes $i$ and $j$, respectively. If $s' = s_i \sqcup s_j = \bigsqcup \big( K(i) \cup K(j) \big)$ is the merged state, then the decision function $\dec{}$ is monotonic with regard to merging: $\dec{}(s_i) \leq \dec{}(s')$ and $\dec{}(s_j) \leq \dec{}(s')$.
\end{lemma}
\begin{proof}
Assume actions and processes as above, with a state diagram as shown in \cref{fig:monotonic-stability}.
We first show monotonicity for merging two actions that were concurrently applied in the same state $s_0$ (on different processes). That is, $s_1 = s_0 \sqcup a_1(i)$, and $s_2 = s_0 \sqcup a_2(j)$, and $s' = s_1 \sqcup s_2 = s_1 \sqcup a_2(j)$ with $s_0 \leq s_1 \leq s'$ and $s_0 \leq s_2 \leq s'$.
Because the preconditions for $a_1(i)$ and $a_2(j)$ hold in $s_0$, that is, $p_1(i, s_0) \land p_2(j, s_0)$, we can conclude from stability that also $p_1(i, s_2)$ holds (in green in \cref{fig:monotonic-stability}).
From action monotonicity, we know that $\dec{}(s_0) \leq \dec{}(s_2)$, and also, due to the green precondition $p_1(i, s_2)$, we can now also conclude from monotonicity that $\dec{}(s_2) \leq \dec{}(s')$ because of the blue edge applying $a_1(i)$ from $s_2$ to $s'$.
We can use the same argument to show that $\dec{}(s_1) \leq \dec{}(s')$, thus merging of two states that are the result of concurrent actions in the same starting state is a monotonic step.

This argument extends to arbitrary reachable states, assuming that all processes started in the same initial state. This is because for any concurrent action $a_3(k)$ with $k \neq i$ applied in $s_2$ resulting in state $s_3$ (the top right state in the figure), we can again conclude from stability that $p_1(i, s_3)$ holds (again in green).
Generally, even though stability only requires reasoning about applying actions in the same state $s_0$, by repeatedly applying it to $s_2$, $s_3$, and so on, it does guarantee that $p_1(i, s)$ holds for any state $s$ that is concurrent with $s_1$, that is, where $s_0 \leq s \not\leq s_1$.
Thus, merging $s_1$ into any of its concurrent states is the same as applying $a_1(i)$ to that state, which, according to action monotonicity, preserves the monotonicity of the decision function.
\end{proof}

\begin{figure}
\centering
\begin{tikzpicture}[every node/.style={font=\itshape\tiny}, scale=1.5]

  \node[minimum width=1.3cm, align=center] (S0) at (0,0)  {$s_0$};
  \node[minimum width=2cm, align=right] (Assume) [right of=S0, xshift=2em] {$p_1(i, s_0) \land p_2(j, s_0)$};
  \node[minimum width=1.3cm, align=center] (S2) at (1,1) {$s_2$};
  \node[minimum width=2cm, align=right, color=green] (Implies12) [right of=S2] {$p_1(i, s_2)$};
  \node[minimum width=1.3cm, align=center] (S1) at (-1,1)  {$s_1$};
  \node[minimum width=1.3cm, align=center] (Sm12) at (0,2)  {$s' = s_1 \sqcup s_2$};
  \node[minimum width=1.3cm, align=center] (S3) at (2,2)  {$s_3$};
  \node[minimum width=2cm, align=right, color=green] (Implies23) [right of=S3] {$p_1(i, s_3)$};
  \node[minimum width=1.3cm, align=center] (Sm23) at (1,3)  {$s'' = s_1 \sqcup s_3$};

  \path[->] (S0) edge node[left]{$a_1(i)$} (S1);
  \path[->, dotted] (S1) edge node[left]{$a_2(j)$} (Sm12);
  \path[->] (S0) edge node[right]{$a_2(j)$} (S2);
  \path[->, blue] (S2) edge node[right]{$a_1(i)$} (Sm12);
  \path[->] (S2) edge node[right]{$a_3(k)$} (S3);
  \path[->, dotted] (Sm12) edge node[left]{$a_3(k)$} (Sm23);
  \path[->, blue] (S3) edge node[right]{$a_1(i)$} (Sm23);
\end{tikzpicture}

\caption{From stability and monotonicity of individual steps, follows monotonicity of merges.}
\label{fig:monotonic-stability}
\end{figure}
\begin{theorem}\label{theorem:agreement-monotonicity}
Assume a PRDT with a decision function $\dec{}$ and an initial knowledge state $K$ where $K(i) = \emptyset$ for all $i \in \textit{Uid}$.
Given action monotonicity (\cref{def:action-monotonicity}) holds for all protocol actions and all preconditions are stable (\cref{def:stability}), 
this system guarantees consensus safety as defined in \cref{def:agreement}.
\end{theorem}
\begin{proof}
Following Theorem~\ref{theorem:step-monot-is-agreement}, we know that consensus safety~(\cref{def:agreement}) follows from step monotonicity: $K \rightarrow K' \implies \dec{} \big( \bigsqcup K(i) \big) \leq \dec{} \big(\bigsqcup K'(i) \big)$ for any process identifier $i \in \textit{Uid}$.
According to the semantics defined in \Cref{fig:rules}, 
there are two possible steps that the system can perform: 

(i) Apply an action $a$ at process $i$ (rule \textsc{apply}): in this case, it holds that $\dec{}(\bigsqcup K(i)) \leq \dec{}(\bigsqcup K'(i))$ due to \emph{action monotonicity}.

(ii) Merge the deltas of process $j$ into the state of processes $i$ (rule \textsc{merge}): in this case, it holds that $\dec{}(\bigsqcup K(i)) \leq \dec{}(\bigsqcup K'(i))$ due to \emph{merge monotonicity} (Lemma~\ref{lemma:merge-monotonicity}).
\end{proof}

\paragraph{Discussion}
Ultimately, our verification method requires two conditions for every PRDT implementation: 
(1) \emph{action monotonicity} and (2) \emph{stability of preconditions}. 
These conditions mirror two core verification conditions in the verification method for 
consensus protocols modelled as operation-based replicated data types, as formulated by
 \citet{lewchenkoBoltOnStrong2025} in concurrent work.
Specifically, action monotonicity corresponds to their \emph{decision monotonicity},
while stability of the precondition corresponds to their \emph{racing state stability}.

However, the method of Lewchenko et al. imposes two additional conditions: 
(3) \emph{initiation safety}
and (4) \emph{commutativity of system actions}.
Initiation safety ensures that actions defined in the program are only applied when the preconditions specified in the external protocol definition are met. In the PRDT model, preconditions are intrinsically defined as part of the program. 
This design choice eliminates the need for explicit verification of initiation safety 
because the preconditions are automatically extracted and enforced directly from the program code.
Similarly, establishing commutativity of system actions becomes unnecessary 
due to the strong eventual consistency properties inherent to PRDTs.

Beyond reducing verification obligations for basic protocols such as Voting, the compositional structure of PRDTs further reduces verification complexity. 
Correctness proofs for basic protocols directly support verification of complex composed protocols. For instance, Paxos, which utilizes voting as a building block, can reuse correctness proofs for Voting (cf.  \cref{sec:implementation}).

\subsection{Verifying the Example Voting Protocol}
\label{sec-verifying-voting}

To illustrate the verification process introduced in the previous subsection, 
we verify consensus safety of the example voting protocol.
We provide a high-level proof of a more complicated consensus protocol (Paxos) in Appendix~\ref{sec:paxos-proof}.
\begin{proposition}
The Voting PRDT presented in \cref{fig:voting-rdt} guarantees \emph{consensus safety}.
\end{proposition}
\begin{proof}

As highlighted in the previous subsection, to verify consensus safety, we have two proof obligations: action monotonicity and precondition stability.
The \texttt{Voting} protocol defines a single action, \texttt{voteFor} action (Line~\ref{line:voting-votefor}) guarded by the precondition \texttt{hasNotVoted} (Line~\ref{line:voting-hasvoted}).

The action has a stable precondition, because
processes can only vote using their own unique \texttt{replicaId}, ensuring that any two concurrent \texttt{voteFor} actions will necessarily use different identifiers. This design implies that the \texttt{hasNotVoted} precondition remains stable under concurrent updates, as it checks for the existence of a vote by the local identifier -- a state that cannot be altered by votes from other, distinct identifiers.

To prove action monotonicity of \texttt{voteFor},
we distinguish two cases that could occur when executing \texttt{voteFor} while its precondition is enabled:
\begin{enumerate}
    \item If the decision was previously \texttt{Undecided}, monotonicity holds trivially because every decision result is greater or equal to \texttt{Undecided} w.r.t. the partial order of decisions.
    \item If the decision was previously \texttt{Decided(a)} for any value \texttt{a}, we know that there exists a majority of votes for \texttt{a}. Since the precondition \texttt{hasNotVoted} implies that every process can only vote once, and the set of processes is fixed, we know that there can never be a different majority. Therefore, the decision remains \texttt{Decided(a)}, which satisfies action monotonicity.
\end{enumerate}
\end{proof}

\subsection{Automating the Verification}
\label{sec-automating-verification}

\begin{figure}
\begin{lstlisting}[language=myscala, basicstyle=\scriptsize\ttfamily, xleftmargin=.1\textwidth, xrightmargin=.1\textwidth, escapeinside={//§}{\^^M}, caption={Excerpt of the consensus trait.}, captionpos=b, label={fig:consensus-typeclass}]
trait Consensus[A] {
  def decisionOrder(s1: State, s2: State): Boolean //§\label{line:decisionOrder}

  @law//§\label{line:law}
  def monotonicity(n: Uid, a: Action, s: State): Boolean = {//§\label{line:law-decision-monotonicity}
    val delta = apply(s, a, n)
    val s1 = merge(s, delta)
    precondition(s, a, n) ==> decisionOrder(s, s1)
  }

  @law
  def stability(n1: Uid, n2: Uid, a1: Action, a2: Action, s: State): Boolean = {//§\label{line:law-stability}
    val delta: Delta = apply(s, a2, n2)
    val s1 = merge(s, delta)
    (n1 != n2 && precondition(s, a1, n1) && precondition(s, a2, n2)) ==>
      precondition(s1, a1, n1)
  }
}
\end{lstlisting}
\end{figure}

While we used the Voting PRDT to illustrate manual verification, we employ \emph{Stainless}~\cite{hamzaSystemFR2019, epfllaraStainlesFormal}, a verification framework~\cite{hamzaSystemFR2019, epfllaraStainlesFormal}
for Scala programs, for automated support. Specifically, we map the required \emph{precondition stability} and \emph{action monotonicity} 
verification obligations from \cref{sec:proof-obligations}
to algebraic laws defined in a \texttt{Consensus} trait in \cref{fig:consensus-typeclass}.
Stainless 
automatically checks these laws against developer-provided implementations of this trait. 
For example, for automated verification \texttt{Voting} from \cref{fig:voting-rdt} reduces to implementing 
the \texttt{Consensus[Voting]} type.

The main requirement when instantiating the \texttt{Consensus} trait is to define an order on the PRDT’s decision domain via \texttt{decisionOrder} (Line~\ref{line:decisionOrder}). 
\emph{Action monotonicity} is stated relative to \texttt{decisionOrder} (Line~\ref{line:law-decision-monotonicity}) and encoded as a Stainless \texttt{@law} (Line~\ref{line:law}) that must be discharged for every instance. The same applies to \emph{stability} (Line~\ref{line:law-stability}). In principle, no additional developer input is required; in practice, however, automated verification often needs auxiliary lemmas and carefully placed assertions to guide the prover to a successful proof.

Our verification method effectively exploits the algebraic compositionality of the PRDT model as Stainless can reuse proven component protocols when verifying composite protocols. 
For example, Stainless is able to verify the monotonicity of \texttt{ParallelVoting} 
(\cref{sec:composing-rdts}) 
by relying solely on the previously verified properties of \texttt{Voting}. 
Similarly, we can leverage the monotonicity of a component PRDT to establish precondition stability in composite PRDTs.
For instance, because a \texttt{Voting} decision is monotonic, any precondition waiting for a specific \texttt{Decision(a)} is inherently stable.
We exploit this insight in the next section where we showcase composite Paxos-style PRDTs and reuse the verified properties of the \texttt{Voting} PRDT while 
verifying the Paxos PRDT (see Appendix \ref{sec-verifying-paxos-appendix}).

\section{Case Study: Composition and Decomposition of Paxos-Like Protocols}\label{implementing-a-real-protocol}\label{sec:implementation}

In this section, we demonstrate that the PRDT model effectively scales to real-world consensus protocols by implementing protocols inspired by the Paxos ~\cite{lamportParttimeParliament1998} family. This family encompasses various optimizations and specialized variants developed over decades of research, reflecting both theoretical advancements and practical refinements~\cite{boichatDeconstructingPaxos2003, chandraPaxosMade2007, lamportFastPaxos2006, lamportGeneralizedConsensus2005, lamportPaxosMade2001, nawabDPaxosManaging2018, padonPaxosMade2017, vanrenessePaxosMade2015, maoMenciusBuilding2008, rivettiStateBased2013}. 

We start by describing a complete implementation of classic Paxos, showcasing how complex PRDTs come to life by combining reusable building blocks -- like voting mechanisms, leader election, and shared resources -- each implemented as standalone PRDTs (\cref{sec:genpaxos}).
Afterwards, \cref{sec:multi-round} highlights the power of the RDT-based model in enabling \enquote{adaptable protocols}.
Here, we demonstrate how protocol components can be enhanced with new capabilities, such as transforming single-round protocols into multi-round versions.
\Cref{sec:optimizations} shows how to alter the behaviour of protocol parts beyond composition and how this can be used to implement optimizations such as Multi-Paxos.
Finally, in \cref{sec:reconf}, we tackle dynamic reconfiguration of consensus participants, leveraging our Paxos implementation and the composability of PRDTs.
In addition, in Appendix~\ref{sec-verifying-paxos-appendix}, we provide a high-level proof of the Paxos PRDT, similar to the proof of \texttt{Voting} in \cref{sec-verifying-voting}. A stainless encoding of Paxos will be included as part of the artifact submission.

\subsection{Paxos as a Data Type}
\label{sec:genpaxos}
\label{sec:paxos}

The example voting protocol 
from \cref{sec:overview} 
can stall when each process votes for a different value. Consensus protocols used in practice such as Paxos~\cite{lamportParttimeParliament1998} address this by allowing multiple voting rounds consisting of two phases: first, processes elect a leader; second, the leader proposes a value. If a majority accepts, this value becomes the final decision. If either phase halts, processes terminate the attempt and start a new round.

\paragraph{Involved Components}
The listing below shows the data types of our Paxos PRDT.
A Paxos PRDT is modelled as a map of ballot numbers to corresponding \texttt{PaxosRound} values.
A ballot number is a unique identifier that is used to order and identify Paxos rounds.
Each Paxos round is represented as a product of \texttt{leaderElection} and \texttt{proposals}, 
which are both
instantiations of the \texttt{Voting} PRDT that was introduced in \cref{fig:voting-rdt}.
This composition forms a new data type
that encapsulates a complete consensus round.
The type parameter \texttt{A} represents the values that the consensus protocol seeks to agree upon. For example, in a distributed key-value store, these could be operations such as \texttt{write} or \texttt{read}.
\begin{lstlisting}[language=myscala]
case class Paxos[A](rounds: Map[BallotNum, PaxosRound[A]])
case class BallotNum(uid: Uid, counter: Long) 
case class PaxosRound[A](leaderElection: LeaderElection, proposals: Voting[A])
type LeaderElection = Voting[Uid]
\end{lstlisting}

Instead of organizing rounds as a sequential list of \texttt{PaxosRound} instances, we use a map-based structure where each round is uniquely identified by a monotonically increasing \emph{ballot number}. This design enables efficient representation of PRDT delta-states, allowing any arbitrary subset of ballot numbers to be captured.
Similar to other Paxos implementations~\cite{howardGeneralisedSolution2019}, we use a design where only one process can attempt to become the leader for a given round. This constraint is enforced by deriving each \texttt{BallotNum} from the process ID and a counter. 
\texttt{BallotNum}s are totally ordered by comparing their counters first and using lexicographic ordering of process IDs as a tiebreaker.
\begin{figure}
\begin{lstlisting}[language=myscala, basicstyle={\scriptsize\ttfamily}, caption={Paxos implemented as a PRDT.}, captionpos=b, label={fig:paxos-prdt}]
case class Paxos[A](
    rounds: Map[BallotNum, PaxosRound[A]] = Map.empty[BallotNum, PaxosRound[A]],
    participants: Set[Uid]
):
  // voting:
  def voteLeader(round: BallotNum, leader: Uid)(using LocalUid): PaxosRound[A] = //§\label{line:paxos-voting}
    PaxosRound(leaderElection = rounds.(round).leaderElection.voteFor(leader))
  def voteValue(round: BallotNum, value: A)(using LocalUid): PaxosRound[A] =
    PaxosRound(proposals = rounds.getOrElse(round,PaxosRound()).proposals.voteFor(value))

  // preconditions//§\label{line:paxos-preconditions}
  def roundHasCandidate(round: BallotNum, leaderCandidate: Uid): Boolean =
    // check if given round has candidate
    rounds.get(round) match
      case Some(r) => r.leaderElection.votesMap.values.contains(leaderCandidate)
      case _ => false
  def isLeader(round: BallotNum)(using LocalUid): Boolean =
    // only the leader can start the value voting by proposing a value
    rounds.get(round) match
      case Some(r) =>
        (r.leaderElection.decision == Decided(replicaId)) &&
        r.proposals.votesMap.isEmpty
      case _ => false
  def hasProposal(round: BallotNum, value: A)(using LocalUid): Boolean = 
    // check if the value we are voting for was proposed
    rounds.get(round) match
      case Some(r) => r.proposals.votesMap.values.contains(value)
      case _ => false

  // protocol actions:
  def phase1a(round: BallotNum)(using LocalUid): Paxos[A] = //§ \label{line:phase1a}
    // try to become leader
    precondition(round.uid == replicaId){
      Paxos(Map(round -> voteLeader(replicaId)))
    }
  def phase1b(round: BallotNum, leaderCandidate: Uid)(using LocalUid): Paxos[A] = //§ \label{line:phase1b}
    // vote in the current leader election
    precondition(roundHasCandidate(round, leaderCandidate)){
      Paxos(Map(currentBallotNum -> voteLeader(leaderCandidate)))
    }
  def phase2a(round: BallotNum)(using LocalUid): Paxos[A] = //§ \label{line:phase2a}
    // propose a value if I am the leader
    precondition(isLeader(round)){
      if newestReceivedVal.nonEmpty then
        // propose most recent received value
        Paxos(Map(currentBallotNum -> voteValue(newestReceivedVal.get)))
      else
        // no values received during leader election, propose my value
        Paxos(Map(currentBallotNum -> voteValue(myValue)))
    }
  def phase2b(round: BallotNum, value: A)(using LocalUid): Paxos[A] =  //§\label{line:phase2b}
    // accept proposed value
    precondition(hasProposal(round, value)){
      PaxosDelta(Map.empty + ((voteValue(round, value))))
    }

  // decision function
  def decision: Decision[A] = //§\label{line:paxos-decision}
    rounds.collectFirst {
      case (b, PaxosRound(_, proposals))
      if proposals.decision != Undecided => proposals.decision
    }.getOrElse(Undecided)
\end{lstlisting}
\end{figure}

\paragraph{Putting the Pieces Together}

\cref{fig:paxos-prdt} displays the Paxos PRDT in Scala.
In Line~\ref{line:paxos-voting}, we define two helper methods, \texttt{voteLeader} and \texttt{voteValue}, which cast votes in the
\texttt{leaderElection} and \texttt{proposals}, respectively. 
This demonstrates how composite data types, such as \texttt{Paxos}, can expose protocol actions from their inner data types (in this case, \texttt{Voting}).
Importantly, these methods inherit the safety guarantees of the original \texttt{Voting} PRDT: processes cannot vote more than once. Even if the composite type (\texttt{Paxos}) were to misuse these methods by attempting multiple votes, only the first vote would affect the outcome.

In Line~\ref{line:paxos-preconditions}, we define three preconditions which are used by the protocol actions.
We explain these alongside their respective protocol actions, which align with the traditional phases of the Paxos protocol:

In \textbf{phase 1a} (Line~\ref{line:phase1a}), a process attempts to become the leader by initiating a new voting round.
As a precondition, we require that processes can only try to become the leader in rounds that are associated with their process ID.
When inserting the new entry, the \texttt{voteLeader(candidate)} function produces a new \texttt{PaxosRound} with one vote for \texttt{candidate} in the \texttt{leaderElection} and no votes in \texttt{proposals}.
When multiple processes attempt to initiate a new round concurrently, 
all rounds except the one with the highest ballot number are eventually abandoned. 

In \textbf{phase 1b} (Line~\ref{line:phase1b}), processes confirm the leader by casting a vote.
As a precondition, we require that the round we are trying to vote in actually has a running leader election for the given candidate.

\textbf{Phase 2a} (Line~\ref{line:phase2a}) starts whenever a process is confirmed as a leader.
This is determined by the precondition \texttt{isLeader}, which evaluates the decision made by the \texttt{LeaderElection} for the given round.
Additionally, the precondition ensures that there is no value proposed for the current round yet.
Interestingly, this second condition would not be stable on its own but by combining it with a monotone decision from the LeaderElection it becomes stable.
The leader selects a value to propose by examining all prior rounds and
choosing the proposed value from the most recent round (based on the ballot ID). 
Thanks to phase1b, it is guaranteed that this value 
represents the latest information known by any process that participated in confirming the leader.
If no such value exists, the leader is free to propose any value (expressed by selecting \texttt{myValue}).

In \textbf{phase 2b} (Line~\ref{line:phase2b}), processes accept the proposed value by casting their vote in \texttt{proposals}.
The precondition \texttt{hasProposal} checks whether the leader has already proposed a value and that this value matches the one that the process is accepting.
This ensures that voters in phase2b can only acknowledge a proposal by the leader and not vote for any other value.

Finally, Line~\ref{line:paxos-decision} shows the \textbf{decision function} of the Paxos PRDT.
This function checks whether any round of proposals lead to a decision and -- if yes -- returns that decision.
Using any decision is safe because Paxos guarantees that every round that makes a decision decides on the same value.
Otherwise, it returns \texttt{Undecided}.

\subsection{Algebraic Composition for Multi-Shot Protocol Design}
\label{compositionality-and-data-type-properties}
\label{sec:multi-round}

The Paxos PRDT 
is \enquote{single-shot}, meaning it is limited to making a single decision at a time.
However, processes often require consensus on a \emph{series} of values, such as transaction ordering. 
In the following, we demonstrate how to construct composed data types that can
manage various types of decisions -- ranging from history-free to totally or partially ordered sequences --
by combining single-shot PRDTs with standard data structures like epochs, 
lists, and graphs.

\paragraph{Scenario 1: Series of decisions without retaining a history.}
Consider a manufacturing scenario where autonomous robots must coordinate task assignments to prevent collisions.\footnote{This scenario is inspired by a case study from Actyx:\\\url{https://2023.splashcon.org/details/plf-2023-papers/11/Local-first-at-Actyx}}
While past decisions can be discarded once a robot begins execution, it is crucial to prevent concurrent task assignments.
\emph{State-based Paxos}~\cite{rivettiStateBased2013} is a variant of Paxos that ensures this.
In the PRDT-based model, we can define this variant by composing the single-shot Paxos PRDT into the \texttt{Epoch} type, creating the \texttt{EpochPaxos} data type. An Epoch is a product consisting of a counter and a value of a generic type \texttt{A}.
\begin{lstlisting}[language=myscala]
case class Epoch[A](counter: Int, value: A)
case class EpochPaxos[A](inner: Epoch[Paxos[A]]):
  def nextDecision() =
    precondition(isDecided(inner.value))(
      EpochPaxos(Epoch(counter + 1), Paxos())
    )
\end{lstlisting}

We create \texttt{EpochPaxos} by instantiating the generic type \texttt{A} with the previously defined \texttt{Paxos} PRDT.
\texttt{Epoch} has a special merge function: 
When merging two instances the one with the larger counter is returned;
if both counters are equal, then the two inner values are merged.
Thus, to start a new \enquote{epoch} and decide on a new value, 
the \texttt{nextDecision} method produces a new state with a larger counter and an empty \texttt{Paxos} instance.
This assumes that the prior decision is no longer needed.

\texttt{EpochPaxos} does not need a new decision function: It reuses the decision function of the inner \texttt{Paxos} data type for each \texttt{Epoch}.
But it does introduce a new protocol action \texttt{nextDecision} with a precondition that checks whether the inner \texttt{Paxos} has passed a decision.
This approach is safe, assuming the \texttt{Paxos} implementation itself is safe. The precondition specifies that an action transitioning to epoch $n+1$ can only be applied once epoch $n$ has been decided. This strict ordering prevents concurrent actions, thereby ensuring stability.
This demonstrates how reusing verified PRDT components facilitates modular reasoning, significantly reducing the verification burden for protocol adaptations.

\paragraph{Scenario 2: Series of decisions forming a totally ordered log.}
A transaction processing system might want to keep all past decisions 
to reason about the validity of past and current decisions.
We can represent a list of decisions by composing a 
standard \texttt{List} ARDT with the \texttt{Paxos} PRDT:
\begin{lstlisting}[language=myscala]
  case class SequencePaxos(log: List[Paxos]):
    def nextDecision() =
      precondition(log.forall(isDecided))(
        SequencePaxos(log.append(Paxos()))
      )
\end{lstlisting}
The composition automatically derives a merge operation for \texttt{SequencePaxos} from the 
merge semantics for lists and \texttt{Paxos}.\footnote{The merge operation for lists used by \texttt{SequencePaxos}
combines two instances by merging entries at the same index of the lists. 
If one list is longer than the other, then the additional index positions are kept.}
\texttt{SequencePaxos} maintains all decided values from each \texttt{Paxos} instance. 
Its decision domain is the list of individual decisions.
A common approach to defining a decision order for such cases is to check if one list is a prefix of the other.
Following this, \texttt{SequencePaxos} is consensus safe if it can only append new decisions to the list without changing existing ones.

Given the implementation, processes can only initiate voting in the $n$th \texttt{Paxos} instance after consensus is reached in all preceding instances $(0, ..., n-1)$.
Building on the consensus safety of \texttt{Paxos}, this ensures consensus safety of \texttt{SequencePaxos} (similar to \texttt{EpochPaxos}), because the precondition prevents appending new entries before the previous ones are decided.
This approach aligns with established protocols like \emph{Mencius}~\cite{maoMenciusBuilding2008}.

\paragraph{Scenario 3: Series of partially-ordered decisions.}

Many real-world applications perform independent operations that do not require strict sequencing. In a shopping system, operations within a single order (like \emph{checkout}, \emph{payment}, and \emph{delivery}) must be ordered, but operations across different orders can execute concurrently for boosting performance. %
\emph{Generalized consensus protocols}~\cite{lamportGeneralizedConsensus2005, whittakerSoKGeneralized2021} address this need by allowing parallel decisions for independent operations.

With the PRDT-based model, we can easily lift a single-shot consensus protocol to a generalized multi-shot protocol by modelling instance dependencies as a graph, as demonstrated by 
the \texttt{GenPaxos} data type below.
The state of \texttt{GenPaxos} instances is modelled as a map from \texttt{Uid}s (in this case, representing individual decisions, not processes) to \texttt{Paxos} instances. Each entry contains both a \texttt{Paxos} data type and a set of predecessor operations that this decision \emph{depends on}.
For example, in a shopping system, predecessors would include prior operations within the same order. 
Processes can vote on multiple independent \texttt{Paxos} instances simultaneously, 
provided their predecessors have reached consensus.

Although \texttt{GenPaxos} is a relatively complex composite type, the merge function of all involved parts is derived automatically: case classes merge their components, sets are unioned together, and maps merge the values at the same keys, keeping key-value pairs only present in either of the merged states.
\begin{lstlisting}[language=myscala]
  case class PaxosWithPredecessors(consensus: Paxos, predecessors: Set[Uid])
  case class GenPaxos(operations: Map[Uid, PaxosWithPredecessors]):
    def nextDecision(predecessors: Set[Uid]) =
      precondition(predecessors.forall(p => isDecided(operations(p).consensus)))(
        GenPaxos(Map(Uid.gen, PaxosWithPredecessors(Paxos(), predecessors)))
      )
\end{lstlisting}
The decision domain of \texttt{GenPaxos} is a map from unique operation identifiers (\texttt{Uid}) to the decisions made by their respective inner Paxos instances. A decision in this context is considered larger or equal if each individual map entry (i.e., each inner Paxos decision) represents a larger or equal decision.
Consensus safety of \texttt{GenPaxos} follows because operations are only added after their predecessors have reached a decision, and because that decision is consensus safe due to the inner Paxos.
The precondition is stable because we assume that \texttt{Uid.gen} produces globally unique IDs, thereby preventing conflicting insertions altogether.

\subsection{Protocol Optimizations}\label{sec:optimizations}

Thus far, we have focused on algebraically composing multiple PRDT components without altering their behaviour. 
However, there are scenarios, where one would like to alter the behaviour of the components being composed, e.g.,
enhance performance in a certain deployment setting~\cite{lamportFastPaxos2006, huntZooKeeperWaitfree2010, kotlaZyzzyvaSpeculative2010, maiyyaUnifyingConsensus2019, maoMenciusBuilding2008, nawabDPaxosManaging2018, vanrenessePaxosMade2015, weiVCorfuCloudScale2017,wangParallelsPaxos2019}.
One common optimization 
for executing multiple Paxos instances sequentially 
is \emph{Multi-Paxos}~\cite{chandraPaxosMade2007, lamportPaxosMade2001}. 
Multi-Paxos assumes a stable leader process, eliminating the need for leader re-election for every Paxos instance.
This allows subsequent instances to skip leader election and directly begin in phase 2, where the leader proposes values.

We can model this behaviour as the following \texttt{MultiPaxos} PRDT.
\texttt{MultiPaxos} is a variation of \texttt{EpochPaxos} from the previous subsection, but it uses a copy of the already decided leader election instead of conducting a new election.
This approach ensures that phase 2 is already enabled, allowing the leader to directly start proposing values.
Reusing \texttt{Voting} is safe because it creates a shortcut to a decision that was already reached before, and it performs this shortcut deterministically for all replicas.

\begin{lstlisting}[language=myscala]
case class MultiPaxos[A](inner: Epoch[Paxos[A]]):
 def nextDecision() =
  precondition(isDecided(inner.value))(
    MultiPaxos(Epoch(counter + 1), Paxos(
           Map(nextBallotNum ->
               PaxosRound(
                 leaderElection = currentLeaderElection,
                 proposals = Voting[A]() // empty voting
               )))))
\end{lstlisting}

\subsection{Reconfigurations}
\label{sec:reconf}
Consensus algorithms like Paxos or Raft require a fixed set of processes (a \emph{configuration}) that determines key parameters like majority thresholds.
However, practical systems need \emph{reconfigurations}~\cite{lamportReconfiguringState2010} to dynamically modify process membership, enabling removal of faulty processes, addition of new participants, and adaptation to changing system requirements. This capability is crucial for maintaining system availability and scalability.
Following the same algebraic design style used for other Paxos variants, reconfigurations can be expressed in the PRDT-based model by composing existing building blocks.
Concretely, by binding each consensus protocol instance to a specific configuration, using the consensus mechanism itself to establish configuration changes. Here is an implementation expressing this pattern as a PRDT:
\begin{lstlisting}[language=myscala]
  case class ConfigurationRound[A](
    currentMembers: Set[Uid]
    nextMembers: Paxos[Set[Uid]],
    innerConsensus: Paxos[A])
  case class ReconfigurablePaxos[A](inner: Epoch[ConfigurationRound[A]]):
    def nextDecision() =
      precondition(
        isDecided(inner.value.nextMembers) &&
        isDecided(inner.value.innerConsensus)
      )(ReconfigurablePaxos(
          Epoch(inner.counter +1, ConfigurationRound(
              inner.value.nextMembers.decision.get,
              Paxos(), Paxos()))))
\end{lstlisting}
Each round decides both a value and the membership for the next decision, enabling reconfiguration support in previously fixed-configuration protocols. The \texttt{ConfigurationRound} type integrates two \texttt{Paxos} instances into a product algebraic data type: \texttt{memberConsensus} for configuration changes and \texttt{innerConsensus} for value decisions.
As in similar examples, consensus safety directly follows from making a decision only when both inner consensus instances have reached a decision.

\subsection{Concluding Remarks}

This section demonstrated how PRDTs enable the modular creation of various consensus variants -- ranging from simple voting to complex Paxos-like protocols and variations thereof to reconfigurable protocols.
The PRDT framework fosters recursive protocol designs starting with PRDTs that implement basic consensus building blocks at the bottom and encompassing ever more advanced PRDTs that extend and compose simple ones.
By extending the existing ARDT library~\cite{kuessnerAlgebraicReplicatedData2023}, PRDTs also enable protocol developers to leverage their language integration with modern programming language features such as parametric algebraic data types, traits and type classes.

Crucially, the model relieves protocol developers from explicitly handling network message exchanges. 
Instead, the built-in convergence guarantees allow developers to focus on individual actions and their preconditions to ensure consensus safety.
When composing smaller components that have been proven safe individually (see \cref{sec-automating-verification}), protocol designers can take the safety of those smaller components for granted when reasoning about the safety of the composed protocol. For example, they can rely on the fact that decisions made by verified components remain final and do not change later on.

\section{Evaluation}
\label{sec:evaluation}

PRDTs enable consensus protocol development at a higher level of abstraction by shifting from message-passing to knowledge-based modelling and supporting modular composition. To investigate their 
comparative performance,
we formulate the following research question:

\begin{enumerate}
\item[] Practical Feasibility: Can PRDTs achieve performance comparable to established consensus implementations in real-world applications?
\end{enumerate}

To answer this question, we implemented a strongly consistent distributed key-value store using an extended version of the \texttt{Multi-Paxos} PRDT introduced in \cref{sec:optimizations} and 
compared its performance against etcd~\cite{etcd}, a widely used key-value store that relies on the Raft consensus protocol. 

\subsection{Our Key-Value Store Implementation}

Our distributed key-value store runs on a cluster of several nodes to ensure fault-tolerance (i.e., the key-value store continues to run, even if some nodes fail). 
The underlying consensus algorithm guarantees that clients always access the most recent version of the data, no matter which node they contact.
More precisely, our implementation guarantees \emph{strict serializability} for all operations -- the same consistency guarantee that etcd provides in its default configuration.\footnote{See \url{https://jepsen.io/analyses/etcd-3.4.3} for a discussion of etcd's guarantees.}
This is achieved through two mechanisms: 
First, all write operations are totally ordered in a grow-only log that is replicated across the nodes (via Multi-Paxos).
Second, read requests are only answered by the current Multi-Paxos leader, guaranteeing that reads always access the latest version of the data.
We use a lightweight PRDT-based heartbeat protocol to track node failures (via timeouts). This ensures that the leader is always connected to a quorum of nodes before answering read requests, thus preventing a \enquote{split brain} situation where two nodes concurrently assume that they are the leader.

Updates are encoded as JSON-serialized deltas and transmitted over direct TCP connections.\footnote{using jsoniter-scala: \url{https://github.com/plokhotnyuk/jsoniter-scala}} 
In our benchmark setup, every client is connected to one server at a time (which is by default the leader server).

\subsection{Experimental Setup}
We evaluated etcd and our key-value store in three different setups: 
\begin{enumerate}
  \item \textbf{One data centre} with a varying number of servers placed on different physical machines colocated in the same data centre. We evaluated the latency and throughput of the two systems with different numbers of clients (increasing concurrent requests) and configurations with 3 and 5 server nodes.\footnote{etcd recommends using no more than 7 nodes, while Google's Chubby lock service recommends 5 nodes: \url{https://etcd.io/docs/v3.6/faq/\#what-is-maximum-cluster-size}}
  \item \textbf{One data centre with leader-failure}, which is a variant of the previous setup, where we let the leader fail after 10 seconds.
  \item \textbf{Geo-replicated data centres} with servers placed on different continents. We distributed nodes across three data centres -- Germany (DE), hosting the leader, Singapore (SG), and the U.S. East Coast (US). We tested two configurations: a minimal setup with one node per region (3 nodes total), and a robust setup with three nodes per region (9 nodes total). The latter is relevant in realistic deployments, because it tolerates node failures within a region while maintaining the site-level redundancy of the former.
\end{enumerate}
Scenario (1) and (3) test the scalability and flexibility of PRDT deployments while scenario (2) tests their fault tolerance.

\paragraph{etcd Configuration}

To assess the overhead that PRDTs introduce on the core consensus path, we need to account for optimizations in etcd that either bypass consensus or reduce its per-operation cost. Concretely, we apply two measures.
First, we benchmark both systems under two workloads: a write-only workload where every request must pass through the consensus protocol, and a (more realistic) read-mostly workload that measures the impact of read optimizations.\footnote{Similar to our implementation, etcd serves reads via a heartbeat-based heuristic rather than invoking the full Raft-based consensus protocol:
  \url{https://deepwiki.com/openshift/etcd/4.2-request-processing\#linearizable-read-processing-readindex}
}
Second, we disable request batching in etcd, which would otherwise merge multiple client operations into a single consensus round and obscure per-operation cost.
A further optimization could not be disabled:
etcd's optimistic pipelining of consensus rounds, which parallelizes concurrent request processing.\footnote{
\url{https://github.com/etcd-io/raft/tree/main?tab=readme-ov-file\#features}}
Unless otherwise noted for the specific scenario, we left all other settings at the default value.

\paragraph{Workload Generation}

We use the Yahoo! Cloud Serving Benchmark (YCSB)~\cite{cooperBenchmarkingCloud2010} to generate workloads and to drive our benchmarks.
Our read benchmarks use workload B (95/5 read/write mix) included with YCSB, while the write benchmarks use a write-only workload to measure consensus performance in isolation.
For the local setups, we use workloads with 100k operations per run while the geo-replicated setup uses 1k operations per run.
We repeat every run 3 times and report averages.
Each run has a fixed number of client threads on the benchmark driver issuing the operations.
Each client waits for an answer before issuing another request, thus, with a single client, all requests are issued sequentially, while with twenty clients, there will be twenty concurrent requests.

\paragraph{Used Hardware and Software}
We use servers with dedicated vCPUs on AMD EPYC hardware with 8 vCPUs and 32 GB of RAM.
The servers run Ubuntu 24.04, etcd version 3.6.7 and OpenJDK 25 (for YCSB and our Scala implementation).%

\paragraph{Round Trip Latencies}
For servers colocated in the same data centre, we measured a round-trip network latency of 0.4 - 0.5 milliseconds.
Inter-region round-trip times (RTTs) were as follows: 161 ms (DE-SG), 112 ms (DE-US), and 216 ms (US-SG).

\subsection{Results}

\paragraph{One Data Center Setup}

\begin{figure}
\centering%
  \includegraphics[width=.65\textwidth]{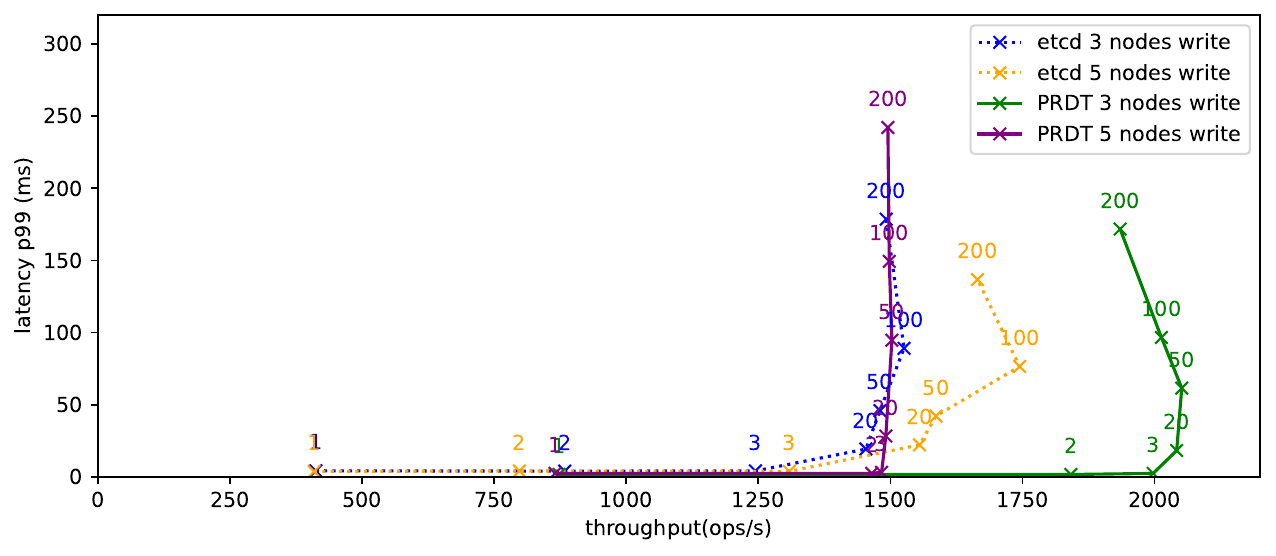} %
\caption{Throughput vs latency for a write-only workload 
in a local setup. 
Each data point is labelled with the number of client threads used to drive the benchmark.
}%
\label{fig:local_write_comparison}%
\end{figure}

\begin{figure}
\centering%
  \includegraphics[width=.65\textwidth]{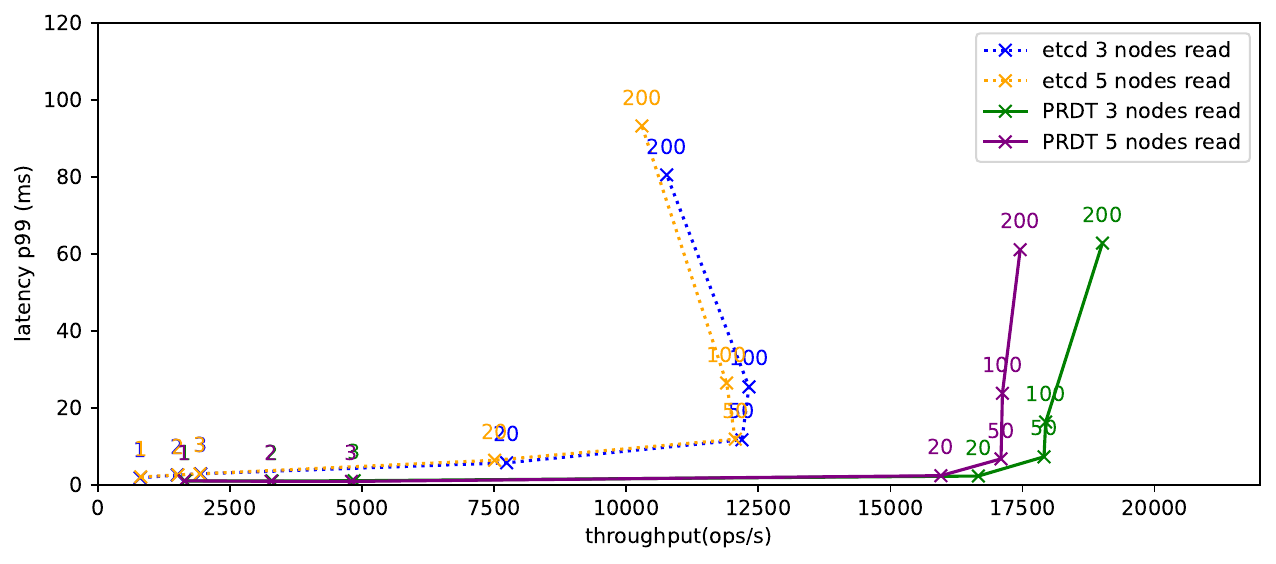} %
\caption{Throughput vs latency for a read-mostly workload (95/5) 
in a local setup.
Each data point is labelled with the number of client threads7used to drive the benchmark.
}%
\label{fig:local_read_comparison}%
\end{figure}

\Cref{fig:local_write_comparison} presents the results for the write-only workload. 
Tables displaying the results in more detail can be found in Appendix~\ref{sec:bench-data}.
With an intra-data centre round trip latency of 0.4 - 0.5 milliseconds, the theoretical maximum throughput is bound to approximately 2,000 - 2,500 requests per second (since write requests need to be processed sequentially).
Both systems approach this limit, indicating that network latency remains the dominant bottleneck, even within a local network.

In the 3-server configuration, our PRDT-based implementation outperforms etcd in both throughput and latency. We attribute this to the relative simplicity of our implementation: etcd's broader feature set likely introduces a modest overhead that becomes visible at this scale.
Scaling to 5 servers reveals a divergence: etcd’s performance improves while the PRDT implementation decreases. 
We hypothesize that etcd makes better use of the larger node pool by exploiting subtle latency differences between servers. 
Consensus throughput is bound by the lowest round-trip latency achievable among any (majority) quorum of nodes, 
and more nodes increase the chances of finding a faster quorum. 
The PRDT implementation, by contrast, is penalized by its simple delta dissemination strategy, which assigns equal priority to all messages regardless of their role in protocol progress. With 5 nodes, the total message volume increases, which presumably delays delivery to the leader -- the critical path for consensus progress -- and consequently reduces throughput.

\cref{fig:local_read_comparison} shows results for the read-mostly workload. 
As expected, both systems achieve significantly higher throughput and lower latency than in the write-only case: 
Under regular heartbeat communication, read requests are served directly by the leader without requiring a full consensus round.
Our implementation outperforms etcd for both 3 and 5 nodes. For etcd, the difference between the two configurations is negligible, whereas the PRDT implementation retains a slight advantage in the 3-node case. Since the read-mostly workload still includes 5\% write requests that must pass through consensus, we attribute this to the same effect observed in the write-only benchmark -- the relative impact is simply less pronounced due to the lower write fraction.	

\paragraph{One Data Center with Leader-Failure}

\begin{figure}
\centering%
  \includegraphics[width=.70\textwidth]{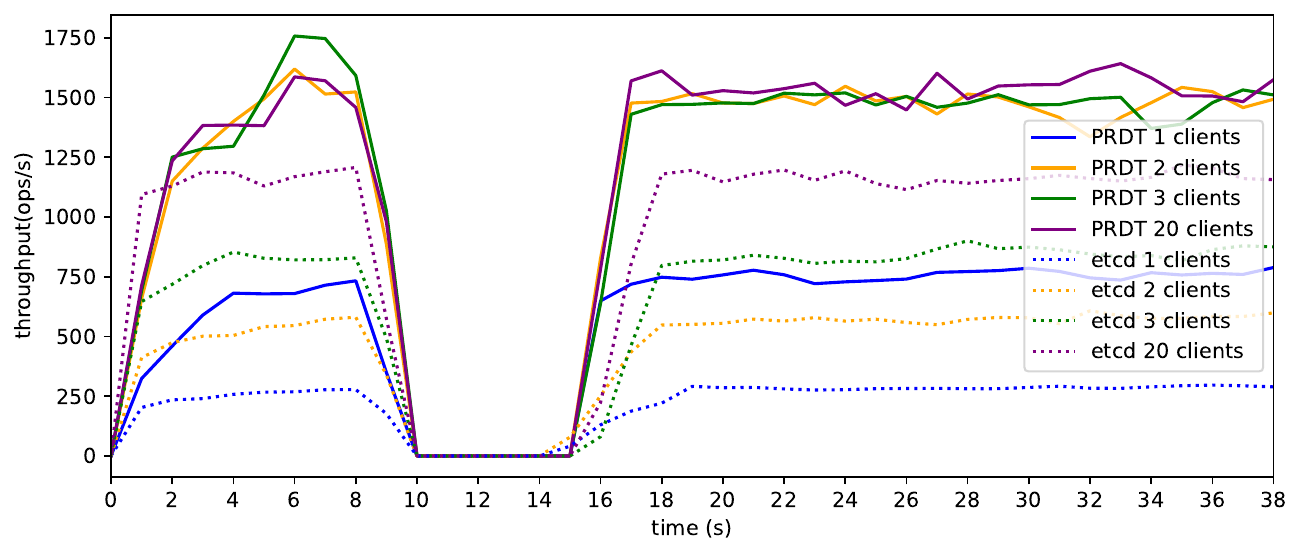}
\caption{Leader failure scenario with 3 servers and different numbers of clients. We terminate the leader after 10 seconds of runtime. Both systems are configured with a 5-second leader timeout.}%
\label{fig:benchmark_leader_failure}%
\end{figure}

Our second data centre experiment evaluates fault tolerance. We configure a 3-server cluster with different numbers of client threads and forcibly terminate the leader node after 10 seconds of runtime. This simulates a crash or network partition where the leader becomes unavailable, causing a temporary halt in progress. For this experiment, both systems are configured with a 5-second leader timeout.
\Cref{fig:benchmark_leader_failure} visualizes the throughput over time. As expected, both systems show a sharp drop to zero throughput immediately following the leader's termination. After approximately 5 seconds (matching the timeout configuration), both systems successfully elect a new leader and resume processing.
In our experiments, etcd took slightly longer than our implementation to recover full throughput. This may partly reflect our client configuration. We used an operation timeout of 1 second for both systems and left all other parameters at their defaults. Tuning these values could likely reduce etcd's recovery time.

\paragraph{Geo-Replicated Setup}

\begin{figure}
\begin{subfigure}{0.40\textwidth}%
\centering%
  \includegraphics[width=\textwidth]{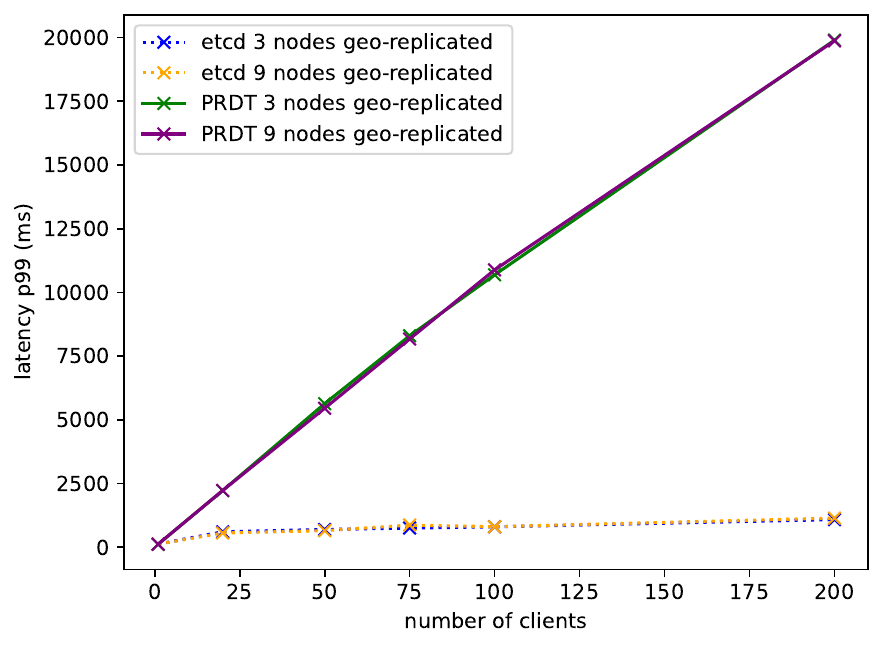} %
\end{subfigure}%
\begin{subfigure}{0.40\textwidth}%
\centering%
  \includegraphics[width=\textwidth]{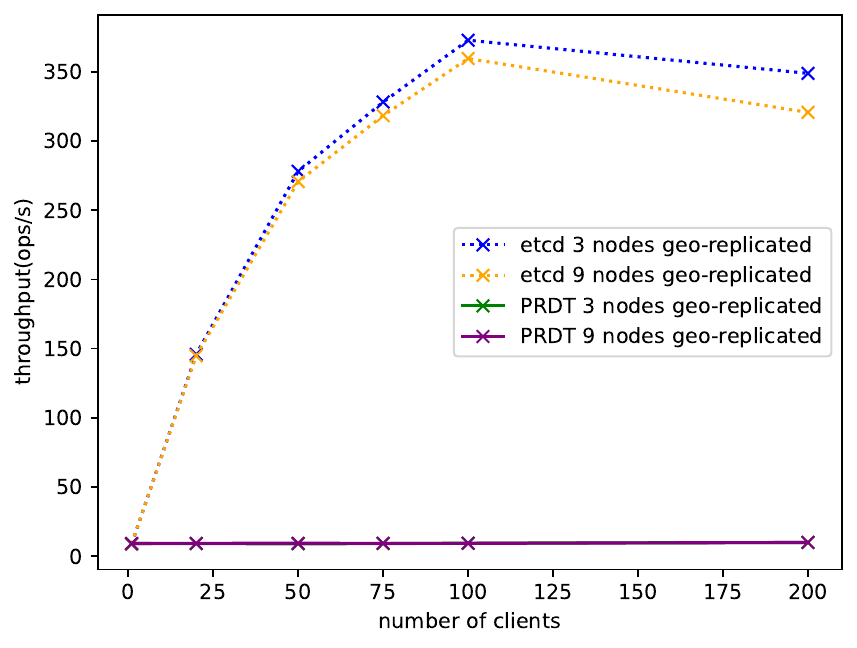} %
\end{subfigure}%
\caption{Latency and throughput with different numbers of clients in a geo-replicated setup with 3 servers placed in Germany (leader), Singapore and U.S. east coast. The clients are located on a different machine in the same data centre as the leader.}%
\label{fig:dist_comparison}%
\end{figure}

In our geo-replicated setup, the benchmark driver was colocated with the leader (on a separate machine in the same data centre) in Germany to isolate consensus latency. %
The read performance in this setting is the same as in the purely local setup, therefore we only report results for the write-only workload.

\Cref{fig:dist_comparison} visualizes throughput and latency in relation to the number of concurrent client requests.
In this high-latency setting, our implementation's performance is dominated by network transit times. Notably, we observed no significant performance difference between the 3-node and 9-node setups -- for throughput the PRDT 3-node graph overlays with the 9-node graph.

In both cases, the leader must communicate with at least one remote region to form a quorum, and this inter-region latency masks the overhead of coordinating additional nodes.

For a single client, our average latency was 120 ms, which aligns with the RTT between Germany and the U.S. This confirms that the consensus speed is bound by the fastest round-trip to a quorum (in this case, involving nodes in Germany and the U.S.).
However, as we increase client concurrency, our sequential implementation hits a bottleneck: additional clients merely extend the request queue at the leader, linearly increasing latency without improving throughput. In contrast, etcd effectively leverages pipelining: While its single-client performance is similarly latency-bound, it achieves higher throughput with multiple clients by overlapping consensus rounds, an optimization that is particularly effective here since our workload generates no conflicts that would force rollbacks.

\paragraph{Concluding Discussion}
In summary, our experiments provide strong evidence for the practical feasibility of the PRDT model.
In both local and geo-replicated settings, throughput and latency were primarily limited by the round-trip latency between consensus nodes, indicating that the PRDT abstraction and its modular composition impose no prohibitive overhead on the core consensus path that would preclude practical adoption. Our leader-failure experiment further confirms that the PRDT-based implementation exhibits the fault-tolerance properties expected of Multi-Paxos.

Visible performance differences between the two systems were confined to the geo-replicated benchmarks, where etcd benefits from request pipelining absent in our prototype. We consider such optimizations compatible with the PRDT model and real-world systems could adopt pipelining and similar strategies without abandoning the PRDT abstraction.
Furthermore, PRDTs might open additional novel optimization avenues: application-specific consensus variants can be composed directly from protocol building blocks, and deployment-specific delta dissemination strategies -- such as prioritizing traffic to the current leader -- could further improve performance. Our current prototype uses the default, generic dissemination algorithm provided by the ARDT framework.
Validating these optimization paths remains future work.

\section{Related Work}
\paragraph{Abstractions and Building Blocks for Distributed Protocols}\label{building-blocksabstractions-for-consensus}

\citet{andersenProtocolCombinators2021} propose combinators as a monadic embedded DSL in Haskell for composing protocols from simpler components. While sharing our compositional approach, their work relies on message passing rather than replicated state, and requires manual specification of monadic interpretation semantics. This adds flexibility but contrasts with our model's automatic guarantees of eventual common knowledge and simplified reasoning.
Another work focussing on composition of protocol components is \emph{QuickSilver}\cite{jaberQuickSilverModeling2021}, a modelling language and automated verifier for compositional agreement protocols. Protocols modelled with \emph{Quicksilver} use message passing, requiring a finer level of granularity than our knowledge-based model using replicated state. In addition, \emph{Quicksilver}'s focus lies in the automated verification of protocol models, which unlike PRDTs does not yield executable protocol implementations.

\citet{howardGeneralisedSolution2019} present a generalized Paxos using the abstraction of infinite write-once registers for communication. While sharing our focus on immutable replicated state but without exploring implementation or compositional aspects, they do not provide an implementation or explore opportunities for data type composition. In follow-up work~\cite{howardPaxosVs2020}, the same authors analyze similarities between Paxos and Raft, examining core principles. \citet{wangParallelsPaxos2019} present a formal mapping between Paxos and Raft, discussing similarities and showing how to port optimizations. 
These works indicate that grounding protocol construction in common core principles and reusable components, as our approach advocates, facilitates reasoning/understanding, and enable broader optimization opportunities.

\citet{boichatDeconstructingPaxos2003} provide a formal model that decomposes Paxos into leader election and register abstractions, while \citet{guptaChemistryAgreement2023} present a chemical notation to categorize agreement protocol components. Though both works identify key protocol building blocks, they remain in the message-passing paradigm and do not provide a programming model for practical protocol composition.

\paragraph{Monotonicity and Lattice-Based Reasoning}\label{monotonicity-and-lattice-based-reasoning}

Several works demonstrate the value of lattices and monotonicity in distributed systems. The \emph{CALM} theorem establishes relationships between program monotonicity and consistency requirements~\cite{amelootRelationalTransducers2013, Alvaro2011, amelootWeakerForms2015,hellersteinKeepingCALM2020, baccaertDistributedConsistency2023}. \citet{laddadKeepCALM2022} extend this to CRDTs, showing how monotonic query functions enable coordination-free replication.
We employ monotonicity in several parts of our model when we reason about the stability of preconditions and the monotonicity of the decision function (cf.~ \cref{sec:verification}). While existing work mainly emphasizes coordination freedom and eventual consistency, we demonstrate how replicated data types can implement strong consistency through eventual common knowledge.

\citet{conwayLogicLattices2012} introduce lattices as explicit language constructs for monotonic distributed programs, but leave lattice implementation correctness to programmers, unlike our ARDT-based model that can derive correct lattices through automatic composition of the merges functions of their components.
A similar idea was later explored by \citet{kuperLVarsLatticebased2013} who propose \emph{LVars}, lattice structures for parallel programming.

\emph{Derecho}~\cite{birmanInvitedPaper2023, jhaDerechoFast2019} is a consensus implementation for data centres which leverages monotonicity to speed up protocol execution.

Prior work~\cite{v.gleissenthallPretendSynchrony2019,farzanSoundSequentialization2022} has proposed various techniques to simplify reasoning about distributed systems by abstracting away message-passing complexity. 
Closest to our approach is recent work by \citet{lewchenkoBoltOnStrong2025} who propose a proof theory for consensus protocols built on top of operation-based replicated data types.
We adapt techniques from their verification approach to the PRDT model in \cref{sec:verification}.
By building on state-based replicated data types and a programming model that features integrated preconditions, we are able to reduce the necessary verification conditions to show consensus safety from four to two (see \cref{sec:proof-obligations}).
Another important difference is that while their work focuses on a proof theory that simplifies verification of one concrete protocol implementation, we describe a compositional approach that allows modular verification and recombination of multiple consensus building blocks.

\paragraph{Consistency Aware Programming Models for Distributed Systems}\label{programming-models-for-distributed-systems-and-programs}

Several languages and models integrate consistency reasoning into distributed programming~\cite{houshmand2019,Alvaro2014, Lewchenko2019,kohlerRethinkingSafeConsistency2020,Balegas2015,Balegas2018,Milano2019}, offering both fast but weak consistency and slower but strong consistency modes for operations in a program. While \citet{haasLoReProgrammingModel2024} and \citet{deporre2021} enable programmable weak consistency through replicated data types, they rely on external mechanisms with fixed semantics that are optimized towards specific use cases for strong consistency. PRDTs, in contrast, additionally offer strong consistency as replicated data types, enabling customizable protocols for diverse application requirements.

\section{Conclusion and Future Work}

We introduced \emph{protocol RDTs} (PRDTs), a programming model for composable consensus protocols.
PRDTs are an extension to the ARDT framework~\cite{kuessnerAlgebraicReplicatedData2023}, introducing new abstractions for high-level implementations of consensus protocols.
PRDTs implement a knowledge-based model of consensus, 
where local state is treated as a representation of eventual common knowledge in the system.
This relieves protocol designers from low-level concerns like message delivery and reordering, allowing them to focus their attention on the protocol logic and how to evolve the state of knowledge such that exactly one value is decided for every consensus instance.
This implementation and reasoning style is made possible by the fact that PRDTs 
automatically ensure convergence of the \enquote{knowledge state}.

We demonstrated that our model enables building complex protocols like Paxos through composition of basic consensus building blocks.
Moreover, we showed how to use composition and generalized higher-order components: 
This allows refining protocols further by lifting them to new protocols with added capabilities such as multiple decisions in a row, and to define advanced protocols by composing simpler ones. 
Finally, we demonstrated that building protocols as PRDTs does not introduce inherent unacceptable performance overhead: We did so by empirically comparing a prototype key-value store implementation based on our PRDT-based Multi-Paxos with the well-established key-value store implementation 
etcd~\cite{etcd}.
In our experiments, both systems demonstrated comparable performance.

We see two promising directions for future work.
First, the monotonicity and eventual convergence of PRDTs creates opportunities for protocol optimizations.
By taking monotonicity for granted, protocols can potentially \enquote{cut corners} and loosen certain ordering requirements, which would be necessary in a traditional message passing setting~\cite{birmanInvitedPaper2023, jhaDerechoFast2019,chuOptimizingDistributed2024}.
Second, PRDTs could enhance mixed-consistency programming models. While existing languages use CRDTs for weak consistency~\cite{haasLoReProgrammingModel2024,deporre2021} and external systems for strong consistency, PRDTs could unify both consistency levels under replicated data types.

\section{Data-Availability Statement}

We will submit our PRDT library as well as our Stainless verified protocol implementations (implemented in Scala) for artifact evaluation.
Additionally, the artifact will include our performance evaluation (Scala implementation of the key-value store, datasets of the results, analysis scripts).
All of our implementations are open source under Apache-2.0 licence and will be made available with the non-anonymized version of this work.

\bibliographystyle{ACM-Reference-Format}
\bibliography{bibliography}

\clearpage
\appendix
\section{Verifying the Paxos PRDT}
\label{sec-verifying-paxos-appendix}
\label{sec:paxos-proof}

\begin{proposition}
The Paxos PRDT presented in \cref{fig:paxos-prdt} guarantees consensus safety (\cref{def:consensus-safety}).
\end{proposition}

\begin{proof}

Following Theorem~\ref{theorem:agreement-monotonicity}, we can prove consensus safety by showing action monotonicity and stability of all preconditions.

\emph{Stability of the preconditions:}
\begin{itemize}
    \item \emph{phase1a:} The precondition of phase1a is stable because it only depends on the local replicaId which does not change.
    \item \emph{phase1b:} This precondition is stable because it checks for the existence of a vote in a \texttt{Voting} instance. We know that the map of votes (the knowledge state) for a given Voting can only grow, thus existing votes do not disappear.
    \item \emph{phase2a:} This precondition checks that a given process is the leader for a given round and that the given round has no votes in \texttt{proposals}. The former is stable assuming stability \texttt{Voting} decisions. The latter is stable because phase2a is the only protocol action that could cast a vote in an empty \texttt{proposals} voting and phase2a is bound to a concrete process ID through the leader election.
    \item \emph{phase2b:} Similar to the precondition of phase1b, this precondition checks for the existence of a vote in a \texttt{Voting} instance. Again, this is stable because the map of votes can only grow.
\end{itemize}

\emph{Action monotonicity:}
\begin{itemize}
    \item \emph{phase1a}: This protocol action casts a vote in a leader election. By definition of the decision function, this cannot change the decision because the decision function ignores leader elections.
    \item \emph{phase1b}: Analogous to phase1a.
    \item \emph{phase2a}: This protocol action casts a vote in a previously empty proposal. This could only lead to a decision if there is only one participating process.
    In this case, we can assume that the one and only process has seen all previous votes. Therefore, if there was a previous decision, \texttt{newestReceivedVal} would return that decided value and phase2a would propose it again.
    \item \emph{phase2b}: This protocol action casts a vote in a running \texttt{proposals} round.
    By the precondition of phase2b, we know that there has to be an existing vote in \texttt{proposals} and that this action is now voting for that same vote.
    Since the preconditions of phase2a and phase2b ensure that only the leader can cast a vote in an empty \texttt{proposals}, we know that the first vote will always come from some leader process $p$ and every vote after that one will be for the same value.
    We can distinguish two cases:
    \begin{enumerate}
        \item If the decision was previously \texttt{Undecided}, monotonicity holds trivially, even if executing phase2b would increase the decision.
        \item If the decision was previously \texttt{Decided(a)} for any value $a$, this means that the leader process $p$ must be aware of at least one previous vote that led to $a$. This is due to the fact that when $p$ was elected as leader (received the deltas containing the leader election votes), it would also have received the deltas for any preceding voting rounds. To become leader, $p$ must have received deltas from at least a majority of processes, meaning that these deltas must have included at least one vote for $a$ which was also voted for by a majority.

        Therefore, we can conclude that $p$ knew of the value $a$ when it started the current voting round through phase2a. As a result, it must have proposed $a$ as its \texttt{newestReceivedVal} in phase2a and phase2b can only cast a vote for $a$ again.
        Voting for a previously decided value can only ever lead to the same decision and we are done.
        By induction, this argument holds, even if there were multiple other voting rounds after the decision of $a$, as any of the leaders in these rounds would know of $a$ and would thus be forced to propose it again.
\end{enumerate}

\end{itemize}
\end{proof}

\clearpage
\section{Detailed Benchmark Data}
\label{sec:bench-data}

\subsection{Single Data Center Setup}

\begin{figure}[h]
\footnotesize
\begin{tabular}{rrrrrr}
\toprule
nodes & threads & throughput(ops/s) PRDT & throughput(ops/s) etcd & write p99 latency (ms)& write p99 latency (ms)\\
&  & & & PRDT & etcd \\
\midrule
3 & 1 & 872.86 & 412.58 & 1.84 & 4.29 \\
3 & 2 & 1841.61 & 883.59 & 1.74 & 4.07 \\
3 & 3 & 1997.22 & 1244.61 & 2.48 & 4.19 \\
3 & 4 & 1936.06 & 1251.10 & 3.19 & 5.68 \\
3 & 5 & 2027.24 & 1345.05 & 3.98 & 6.42 \\
3 & 10 & 1931.99 & 1504.89 & 9.01 & 10.13 \\
3 & 20 & 2041.75 & 1453.60 & 18.36 & 19.36 \\
3 & 50 & 2051.34 & 1480.15 & 61.60 & 45.95 \\
3 & 100 & 2012.87 & 1525.77 & 96.63 & 89.14 \\
3 & 200 & 1934.73 & 1492.02 & 171.72 & 178.57 \\
5 & 1 & 865.29 & 409.26 & 2.46 & 3.84 \\
5 & 2 & 1464.73 & 797.94 & 2.54 & 4.17 \\
5 & 3 & 1483.24 & 1308.72 & 3.26 & 3.91 \\
5 & 4 & 1430.80 & 1526.73 & 3.74 & 4.09 \\
5 & 5 & 1509.71 & 1626.84 & 4.59 & 4.86 \\
5 & 10 & 1489.34 & 1609.00 & 9.70 & 9.35 \\
5 & 20 & 1491.29 & 1554.98 & 28.36 & 22.23 \\
5 & 50 & 1502.81 & 1586.28 & 94.73 & 42.20 \\
5 & 100 & 1497.79 & 1744.51 & 149.30 & 76.56 \\
5 & 200 & 1495.19 & 1664.89 & 242.11 & 136.73 \\
\bottomrule
\end{tabular}

\caption{Write performance for a write-only workload in a data center local setup where clients and servers are colocated on different machines in the same data center.}
    
\end{figure}

\vspace{-0.1cm}

\begin{figure}[h]
\footnotesize
\begin{tabular}{rrrrrr}
\toprule
nodes & threads & throughput(ops/s) PRDT & throughput(ops/s) etcd & read p99 latency (ms) & read p99 latency (ms)\\
&  & & & PRDT & etcd \\
\midrule
3 & 1 & 1652.29 & 800.43 & 1.01 & 1.83 \\
3 & 2 & 3300.11 & 1502.72 & 1.01 & 2.58 \\
3 & 3 & 4850.97 & 1948.72 & 1.09 & 2.88 \\
3 & 4 & 6674.00 & 2439.60 & 0.76 & 3.21 \\
3 & 5 & 8366.98 & 2903.27 & 0.76 & 3.42 \\
3 & 10 & 15609.06 & 5038.69 & 0.99 & 4.07 \\
3 & 20 & 16665.33 & 7737.47 & 2.30 & 5.74 \\
3 & 50 & 17906.68 & 12190.05 & 7.31 & 11.71 \\
3 & 100 & 17938.34 & 12323.18 & 16.34 & 25.48 \\
3 & 200 & 19012.90 & 10766.04 & 62.81 & 80.52 \\
5 & 1 & 1634.17 & 812.84 & 1.06 & 2.07 \\
5 & 2 & 3278.69 & 1518.64 & 0.94 & 2.61 \\
5 & 3 & 4801.20 & 1933.43 & 0.91 & 2.89 \\
5 & 10 & 14361.88 & 5433.35 & 1.03 & 3.68 \\
5 & 20 & 15955.29 & 7509.22 & 2.40 & 6.41 \\
5 & 50 & 17093.14 & 12063.85 & 6.77 & 11.87 \\
5 & 100 & 17119.45 & 11899.00 & 23.79 & 26.46 \\
5 & 200 & 17455.47 & 10297.13 & 61.10 & 93.24 \\
\bottomrule
\end{tabular}

\caption{Read performance for a read-mostly workload (95/5) in a data center local setup where clients and servers are colocated on different machines in the same data center.}
    
\end{figure}

\subsection{Geo-Replicated Setup}

\begin{figure}[h]
\footnotesize
\begin{tabular}{rrrrrr}
\toprule
nodes & threads & throughput(ops/s) PRDT & throughput(ops/s) etcd & write p99 latency (ms)& write p99 latency (ms) \\
&  & & & PRDT & etcd \\
\midrule
3 & 1 & 9.20 & 8.74 & 114.42 & 117.16 \\
3 & 20 & 9.18 & 145.91 & 2220.08 & 600.45 \\
3 & 50 & 9.02 & 278.09 & 5629.96 & 691.64 \\
3 & 75 & 9.17 & 328.14 & 8293.72 & 748.49 \\
3 & 100 & 9.47 & 372.82 & 10686.43 & 796.51 \\
3 & 200 & 9.98 & 348.87 & 19898.15 & 1081.56 \\
9 & 1 & 9.25 & 8.83 & 114.24 & 116.39 \\
9 & 20 & 9.37 & 144.86 & 2224.29 & 554.15 \\
9 & 50 & 9.39 & 270.67 & 5449.70 & 646.18 \\
9 & 75 & 9.36 & 318.34 & 8169.58 & 863.61 \\
9 & 100 & 9.36 & 359.55 & 10882.02 & 801.57 \\
9 & 200 & 9.98 & 320.68 & 19874.84 & 1139.45 \\
\bottomrule
\end{tabular}

\caption{Write performance for a write-only workload in a geo-replicated setup with servers placed in Germany, Singapore, and U.S. east coast. For the 3 node setup, one node is placed in each data center. For the 9 node setup, 3 nodes are placed in each data center. The benchmark driver running the client threads is colocated with the leader node in the data center in Germany.}
    
\end{figure}

\end{document}